\documentclass[11pt]{article}
\usepackage{dsfont}
\usepackage{environ}
\usepackage{natbib}
\setcitestyle{authoryear,open={(},close={)}}

\usepackage{url}            
\usepackage{booktabs}       
\usepackage{amsfonts}       
\usepackage{microtype}      

\usepackage[usenames,dvipsnames]{xcolor}
\definecolor{Gred}{RGB}{219, 50, 54}
\definecolor{Ggreen}{RGB}{60, 186, 84}
\definecolor{Gblue}{RGB}{72, 133, 237}
\definecolor{Gyellow}{RGB}{247, 178, 16}
\definecolor{ToCgreen}{RGB}{0, 128, 0}
\definecolor{myGold}{RGB}{231,141,20}
\definecolor{myBlue}{rgb}{0.19,0.41,.65}
\definecolor{myPurple}{RGB}{175,0,124}
\definecolor{niceRed}{RGB}{153,0,0}
\definecolor{niceRed}{RGB}{190,38,38}
\definecolor{blueGrotto}{HTML}{059DC0}
\definecolor{royalBlue}{HTML}{057DCD}
\definecolor{navyBlueP}{HTML}{0B579C}
\definecolor{limeGreen}{HTML}{81B622}
\definecolor{nicePink}{RGB}{247,83,148}

\usepackage{cmap}
\usepackage{bm}

\usepackage{amsmath}
\usepackage{amssymb}
\usepackage{amsbsy}
\usepackage{amsthm}

\usepackage[pdftex]{graphicx}

\usepackage{subcaption}
\usepackage{rotating}
\usepackage{float}
\usepackage{tikz}

\usepackage{algorithm, caption}
\usepackage[noend]{algpseudocode}
\usepackage{listings}

\usepackage{enumitem}
\usepackage{makecell}

\usepackage{hyperref}
\hypersetup{
	colorlinks = true,
	urlcolor = {royalBlue},
	linkcolor = {magenta},
	citecolor = {niceRed} 
}
\usepackage[nameinlink]{cleveref}

\Crefname{equation}{Eq.}{Eqs.}

\usepackage{multirow}
\usepackage{array}

\usepackage{chngcntr}

\counterwithin*{equation}{section}
\usepackage{chngcntr}
\usepackage{soul}
\usepackage{dsfont}
\usepackage{nicefrac}
\usepackage{verbatim}

\newtheorem{theorem}{Theorem} 
\newtheorem{assumption}{Assumption} 

\newtheorem{lemma}[theorem]{Lemma}
\newtheorem{claim}[theorem]{Claim}
\newtheorem{proposition}[theorem]{Proposition}
\newtheorem{fact}{Fact}

\newtheorem{question}{Question}

\newtheorem{inftheorem}{Informal Theorem}

\newtheorem{definition}[theorem]{Definition}

\renewcommand{\Pr}{\mathop{\bf Pr\/}}
\newcommand{\E}{\mathop{\mathbb{E}\/}}  

\newcommand{\Cov}{\mathop{\bf Cov\/}}

\newcommand{\poly}{\textnormal{poly}}
\newcommand{\polylog}{\textnormal{polylog}}

\newcommand{\sgn}{\textnormal{sgn}}

\newcommand{\reals}{\mathbb R}

\newcommand{\eps}{\epsilon}

\newcommand{\calG}{\mathcal{G}}

\newcommand{\calN}{\mathcal{N}}

\newcommand{\calT}{\mathcal{T}}

\newcommand{\calX}{\mathcal{X}}
\newcommand{\calY}{\mathcal{Y}}
\newcommand{\calZ}{\mathcal{Z}}

\def\<{\langle}
\def\>{\rangle}

\newcommand{\tv}{\mathrm{TV}}

\DeclareMathOperator*{\argmin}{argmin}

\def\wt{\widetilde}
\def\wh{\widehat}

\def\vec{\bm}

\begin{document}
	
	\title{Differentially Private Regression with Unbounded Covariates}
	\author{
	    \textbf{Jason Milionis}\footnote{\url{jm@cs.columbia.edu}} \\
		\small Columbia University \\
		\and
		\textbf{Alkis Kalavasis}\footnote{\url{kalavasisalkis@mail.ntua.gr}} \\
		\small NTUA \\
		\and
		\textbf{Dimitris Fotakis}\footnote{\url{fotakis@cs.ntua.gr}} \\
		\small NTUA \\
		\and
		\textbf{Stratis Ioannidis}
		\footnote{\url{ioannidis@ece.neu.edu}}\\
		\small Northeastern University
	}
	\maketitle
	\thispagestyle{empty}

\begin{abstract}
  We provide computationally efficient, differentially private algorithms for the classical regression settings of Least Squares Fitting, Binary Regression and Linear Regression with unbounded covariates.
  Prior to our work, privacy constraints in such regression settings were studied under strong a priori bounds on covariates. 
  We consider the case of Gaussian marginals and extend recent differentially private techniques on mean and covariance estimation \citep{gautamHighDimensional, vadhanConfidenceIntervals} to the sub-gaussian regime.
  We provide a novel technical analysis yielding differentially private algorithms for the above classical regression settings.
  Through the case of Binary Regression, we capture the fundamental and widely-studied models of logistic regression and linearly-separable SVMs, learning an unbiased estimate of the true regression vector, up to a scaling factor.
\end{abstract}


\section{Introduction}
\label{sec:introduction}

Ever since the introduction of Differential Privacy (DP) by \cite{dwork_dp}, differentially private variants of statistical estimation procedures have been a research topic of intense interest. The work on learning linear models alone is vast (see \cite{thecostofprivacy_GLMs_cai2020, wang} for two recent reviews).  Empirical Risk Minimization is also the impetus for the development of a broad array of new methods for DP-mechanism design, including output perturbation \citep{towards_practical_dp_convopt, ijcai2017-548, pmlr-v32-jain14}, objective perturbation \citep{ERM_chaudhuri11a, kifer12}, and  gradient perturbation \citep{ERM_bassily_private_2014, abadi2016deep}, to name a few.

Nevertheless, despite the intense interest on this topic, \emph{all of the existing work on regression provides differential-privacy guarantees assuming bounded covariates.} Intuitively, this can be explained by inspecting even the simple least squares estimator used in linear regression. It is easy to see that  estimator's \emph{sensitivity}, i.e., its variability under changes on a single sample, is determined by the design matrix (i.e., the matrix of samples). As sensitivity has a direct effect on differential privacy guarantees, bounding the design matrix's eigenvalues is the prevalent approach for bounding the sensitivity. For this reason, assuming bounded covariates is a ubiquitous assumption in DP literature on both linear regression and learning generalized linear models.

This assumption is quite restrictive, and is frequently identified as a deficiency of DP regression algorithms from a practical standpoint \citep{review_of_amin_neurips}. It is also a significant drawback from a theoretical standpoint, as it precludes studying DP-estimators on data sampled from distributions of \emph{unbounded support}. Even the Gaussian distribution, perhaps the most commonly used generative distribution in statistical machine learning literature \citep{double_descent_2, daskalakis_neurips2020_gaussian, double_descent_1, diakonikolas_robust_linear_regression, nakkiran2019data, kreidler_calculating_2018},
cannot be used in conjunction with the existing DP regression algorithms and maintain DP guarantees. 

Our work aims to directly address this, by providing DP algorithms for regression assuming (unbounded) Gaussian covariates.  In doing so, we leverage and extend the recent work of \cite{gautamHighDimensional}, who proposed differentially private mechanisms for estimating the mean and the covariance matrix of
high-dimensional Gaussian random vectors.

\subsection{Contributions}
Our first major contribution is to answer the following question in the affirmative:
\begin{question}
\label{main_question}
Is private regression analysis with unbounded covariates possible?
\end{question}

We study this problem in the context of three scenarios (see \Cref{section:problem-formulation}): Least Squares Fitting, Binary Regression, and (standard) Linear Regression. In all three, we assume (unbounded) Gaussian covariates.

In the Least Squares Fitting setting, given a training set $\{(\vec X_i, y_i)\}$, our goal is to efficiently and privately compute an estimate that is close to the Least Squares Estimate (LSE), i.e., the coefficients of the best-fitting linear function. In this problem, we assume that \emph{labels} $y_i$  are bounded, but make no further assumptions on how they relate to the covariates $\vec X_i \in \reals^d$. Our main result is the following:
\begin{inftheorem}
\label{inftheorem:inf-result1}
For accuracy $\alpha > 0$ and privacy guarantees $\eps, \delta > 0$, there exists an efficient $(\eps, \delta)$-DP algorithm 
that, with high probability, approximates
arbitrarily $\alpha$-closely
the Least Squares Estimate using
$
n = \wt{O}\left(
d/ \alpha^2 + d^{3/2}\log(1/\delta)/(\alpha\eps)\right)
$
samples.
\end{inftheorem}

In our second setting, Binary Regression,  we further  assume that labels are binary  (i.e., $y_i=\pm 1$) and that covariates are zero mean. Moreover, labels are generated by a generalized linear model of the form  $\Pr[y_i = +1 | \vec X_i] = f(\vec \beta^T \vec X_i)$, where $f:\reals\to[0, 1]$ is the model function and $\vec\beta\in\reals^d$ is the true regression coefficient. This setting captures some of the most fundamental machine learning tasks, such as logistic regression and learning linearly-separable Support Vector Machines (SVMs). Our second main result is that the same differentially private estimator we used in Least Squares Fitting scenario can be applied to Binary Regression to obtain the following guarantees:
\begin{inftheorem}
\label{inftheorem:inf-result2}
For accuracy $\alpha > 0$ and privacy guarantees $\eps, \delta > 0$, there exists an efficient $(\eps, \delta)$-DP algorithm 
that, with high probability, approximates
arbitrarily $\alpha$-closely
the true Binary Regression coefficient up to a multiplicative factor using
$
n = \wt{O}\left(
d/ \alpha^2 + d^{3/2}\log(1/\delta)/(\alpha\eps)\right)
$
samples.
\end{inftheorem}

Finally, we turn our attention to the (standard) Linear Regression setting. Here, labels are given by $
y_i = \vec\beta^T \vec X_i + \eps_i,$
where $\eps_i$ are i.i.d. zero-mean Gaussian noise variables and $\vec\beta\in\reals^d$ is again the true regression coefficient. Note that, in contrast to the two previous settings, labels $y_i$ here are unbounded. Our result follows:
\begin{inftheorem}
\label{inftheorem:inf-result3}
For accuracy $\alpha > 0$ and privacy guarantees $\eps, \delta > 0$, there exists an efficient $(\eps, \delta)$-DP algorithm 
that, with high probability, approximates
arbitrarily $\alpha$-closely
the true Linear Regression coefficient using
$
n = \wt{O}\left(
d/ \alpha^2 + d^{3/2}\log(1/\delta)/(\alpha\eps)\right)
$
samples.
\end{inftheorem}

To the best of our knowledge, these results constitute the first efficient and private algorithms for regression analysis with unbounded feature vectors. From a technical standpoint, our analysis for \Cref{inftheorem:inf-result1} and \ref{inftheorem:inf-result2} relies on the fact that the LSE requires the calculation of the inverse of a moment matrix, as well as the expectation of a central random quantity $y_i \vec X_i$. This latter random quantity had not appeared before in Gaussian mean and covariance estimation procedures, but is key to regression settings. Our main conceptual contribution is that this quantity has sub-gaussian tails, hence, by extending the work of \cite{gautamHighDimensional} and \cite{vadhanConfidenceIntervals} to sub-gaussian vectors, we manage to estimate it in a private and sample-efficient way. Finally, utilizing the above results, we show that we are also able to resolve the fundamental case of Linear Regression with unbounded features, indicated in \Cref{inftheorem:inf-result3}.

\section{Related Work}
\label{section:related-work}

\noindent\textbf{Differentially Private Regression and GLMs with Bounded Covariates.} Linear regression is of course a true workhorse of statistics, and there has been a significant body of  work on the design of computationally and statistically efficient differentially private regression algorithms (see e.g., the recent surveys of \cite{thecostofprivacy_GLMs_cai2020, wang} and the references therein).
Approaches include objective perturbation \citep{towards_practical_dp_convopt, kifer12, ERM_Zhang, ERM_chaudhuri11a},
output perturbation \citep{duchi_neurips_2020, towards_practical_dp_convopt, ijcai2017-548, pmlr-v32-jain14},
gradient perturbation \citep{abadi2016deep, ERM_bassily_private_2014},
subsample-and-aggregate \citep{barrientos2019differentially, dwork2010differential},
and sufficient statistics perturbation \citep{alabi2020arxiv, wang, mcsherry2009differentially}. Additionally, several works study generalizations of such mechanisms to Generalized Linear Models (GLMs) \citep{bayesian_glms, towards_practical_dp_convopt, pmlr-v32-jain14, kifer12}.
Approaches that are used in typical regression settings also include variants of differentially private Stochastic Gradient Descent (DP-SGD) or other form of stochastic convex optimization \citep{DP_SCO_OptimalRates_LinTime,DP_SCO_Optimal1,DP_ERM_Wang_Faster_More_General,ijcai2017-548,abadi2016deep,ERM_bassily_private_2014}, which commonly require the optimization domain to be of bounded diameter.
All above works, thus, either operate under a random setting with bounded covariates, or use a fixed design matrix $X$ with bounded minimum eigenvalue on $X^T X$. Such strong assumptions on the boundedness of feature vectors are precisely the kind of assumptions that our work aims to mend.\\

\noindent\textbf{Mean and Covariance Estimation.} The study of differentially private mechanisms for mean and covariance estimation under bounded covariates is classic (see, e.g., \cite{dp_cov_estimation, dwork2014analyze, mcsherry2009differentially}). \cite{sheffet_DP_OLS_17} studies covariance estimation under Gaussian samples, also applying it to the Least Squares Fitting problem we study here; nevertheless, their differential privacy guarantee assumes an upper bound on  covariates.
\cite{sheffet_2nd_moment_matrix_19} obtains a collection of DP algorithms that approximate the second moment matrix of the given dataset using existing Linear Regression techniques. We remark that, in each provided algorithm, an upper bound on the $\ell_2$ norm of each row of the data matrix $A = [X|\vec y]$ is required. This upper bound does not hold in our Linear Regression setting since the received data (both $X$ and $\vec y$) could be unbounded.
\cite{vadhanConfidenceIntervals} resolve, for the first time, the problem of differentially private univariate Gaussian mean estimation without strong a priori bounds and with almost optimal dependence on problem parameters. Also in the univariate setting, \cite{bun2015differentially}  learn more general distributions w.r.t.~Kolmogorov distance, which is weaker than the total variation considered by \cite{vadhanConfidenceIntervals}; \cite{diakonikolas2015differentially} extend this work to total variation distance, again for univariate distributions.

\cite{gautamHighDimensional} extend the work of \cite{vadhanConfidenceIntervals} to multivariate mean and covariance estimation for high-dimensional Gaussian random vectors -- see \Cref{sec:paremestimation} for a description of their guarantees. Related to our setting, \cite{thecostofprivacy_GLMs_cai2020} provide lower bounds for the sample complexity of differentially-private learning the mean of Gaussian random vectors, though the estimation algorithms they propose operate over bounded covariates. Recently, \citet{AdenAli2021} and \citet{Lydia2021CovarianceAware} studied privately learning multivariate Gaussians from an informational theoretic standpoint; however, no computational methods presently match these sample complexity bounds. The latter underscores  difficulties arising in the unbounded covariates setting.\\

\noindent\textbf{LSE for GLMs.} The differentially private algorithm we propose applies Least Squares Estimation (LSE) to learn the parameters of a binary Generalized Linear Model (GLM) (see~\Cref{theorem:main-result2}) and, more generally, to perform Least Squares Fitting over bounded labels (c.f.~\Cref{theorem:main-result1}). It is well known that, under Gaussian marginals, LSE is an unbiased estimator of the parameter vector of a GLM, up to a scaling factor \citep{muratErdogdu, sun2014learning, brillinger2012generalized}. This is a consequence of Stein's Lemma \citep{liu1994siegel, stein1981estimation} -- see also \Cref{sec:stein}. In the binary setting,  LSE can also be seen as a special case of the Linear Discriminant Analysis (LDA) classification algorithm \citep{friedman2001elements}. Our \Cref{theorem:main-result2} can thus also be seen as a differentially private version of LDA.\\

\noindent\textbf{Concurrent Work.}
There has been vibrant independent and concurrent work to ours on Differential Privacy with connections to high-dimensional statistics \citep{liu2021_mean_estim,liu2021_robust_highdim,hopkins2021_SOS_meanestim,kothari2021_SOS_robustestim,ashtiani2021_gaussians_unbounded,kamath2021_SCO_heavy_tailed,argyris2021_gaussians_unbounded}.
Recent works study the problem of privately learning arbitrary Gaussians
\citep{argyris2021_gaussians_unbounded,ashtiani2021_gaussians_unbounded,kothari2021_SOS_robustestim};
these papers provide (among other things) mean and covariance estimation for arbitrary Gaussians and their techniques can be potentially adopted to extend our results accordingly.
Moreover, \cite{liu2021_robust_highdim} examine various statistical tasks (including linear regression) and propose a novel (but computationally inefficient) algorithm that achieves optimal sample complexity under minimal assumptions for these problems using robust statistics tools (see also \cite{liu2021_mean_estim} for private mean estimation).
The work of \cite{kamath2021_SCO_heavy_tailed} studies differentially private stochastic convex optimization with heavy-tailed data under classical structural assumptions (e.g., smoothness of the loss function and boundedness of the parameter space); their techniques could be applied to regression problems too.
Finally, \cite{hopkins2021_SOS_meanestim} examine the problem of mean estimation under minimal assumptions and pure DP using the framework of Sum of Squares.

\section{Preliminaries}

\noindent\textbf{Notation.}
We use bold fonts for vectors (e.g., $\vec \beta, \vec y)$ and denote the set $ \{1, \dots, n\} $ as $[n]$. When $\vec X_i \in \reals^d$ for $i\in[n]$ are the (random) feature vectors and $y_i \in \reals$ for $i\in[n]$ are the (random) labels of a regression setting, the matrix $X = [\vec X_1 \ \vec X_2 \ \dots \ \vec X_n]^T \in \reals^{n \times d}$ is called the (random) design matrix and the vector $\vec y = [y_1 \ y_2 \ \dots \ y_n]^T \in \reals^n$ is called the (random) response vector.
An extended techical preliminary, with definitions required for our proofs, is in  \Cref{appendix:add-prelim}.\\

\noindent\textbf{Differential Privacy.} We use  standard   $(\eps, \delta)$-DP:
\begin{definition}
[Differential Privacy \citep{dwork_dp}]
\label{def:dp}
A randomized algorithm $M : \calX^n \rightarrow \calY$
satisfies $(\eps, \delta)$-differential privacy (equivalently, is said to be $(\eps, \delta)$-DP) if for every pair of neighboring datasets $X, X'  \in \calX^n$ that differ on at most one element,
\[
\Pr[M(X) \in Y] \leq \exp(\eps) 
\Pr[M(X') \in Y] + \delta
\,, \forall~
Y \subseteq \calY
\, .
\]
\end{definition}

A crucial tool for differential privacy is the adaptive composition theorem, providing the privacy properties of a sequence of algorithms $M_1(X), \ldots, M_N(X)$, where the $i$-th algorithm may depend on the outcomes of the algorithms $M_{1}(X),\ldots, M_{i-1}(X)$, for $i \in [N]$.
\begin{fact}
[Composition of differentially private mechanisms \citep{dwork_dp,dwork2010boosting}]
\label{def:dp-composition}
If $M$ is an adaptive composition of differentially private algorithms $M_1,\ldots, M_N$, where $M_i$ is
$(\eps, \delta_i)$-DP for any $i \in [N]$,
then it holds that $M$ is $(\eps N, \sum_{i=1}^N \delta_i)$-DP and, for every $\delta > 0$,
$M$ is $(\eps\sqrt{6N\log(1/\delta)}, \delta + \sum_{i = 1}^N \delta_i)$-DP.
\end{fact}

\paragraph{DP Gaussian Parameter Estimation.}
\label{sec:paremestimation}
At a technical level, our work 
extends
the tools developed by
\cite{gautamHighDimensional} to privately estimate the mean $\vec \mu$ and 
covariance $\Sigma$ of a $d$-dimensional Gaussian distribution.
Their algorithm, which we call $\textsc{LearnGaussian-hd}$, has the following guarantee:
\begin{theorem}
[Multivariate Gaussian Estimation \citep{gautamHighDimensional}]
\label{theorem:gautam-gaussian}
There exists a polynomial time $(\eps^2/2 + \eps\sqrt{2\log(1/\delta)}, \delta)$-DP algorithm $\textsc{LearnGaussian-hd}$
that takes at least
\[
n = \wt{O} \left( \frac{d^2}{\alpha^2} + \frac{d^2}{\alpha \eps} + \frac{ d^{3/2}\log^{1/2}(\kappa) + d^{1/2} \log^{1/2}(R) }{\eps} \right)
\]
i.i.d. samples $\vec X_i$, $i\in [n]$, from a $d$-dimensional Gaussian $\calN(\vec \mu, \Sigma)$ with unknown mean $\vec \mu \in \reals^d$ and unknown covariance $\Sigma \in \reals^{d \times d}$ satisfying $\| \vec \mu \|_2 \leq R$ and $\mathbb{I}_d \preceq \Sigma \preceq \kappa \mathbb{I}_d,$ and outputs estimates $\wh{\vec \mu}, \wh{\Sigma}$ such that, with high probability, $\tv(\calN(\vec \mu, \Sigma), \calN(\wh{\vec \mu}, \wh{\Sigma})) \leq \alpha$. 
\end{theorem}

We remark that this TV distance bound is implied by the parameter estimation of the mean and covariance matrix in Mahalanobis distance. In short, \textsc{LearnGaussian-hd} produces differentially private estimates of the distribution's parameters using only $\wt{O}(d^2)$ samples. It operates under the following boundedness assumptions for the distributional parameters: 
\[\|\vec \mu\|_2 \leq R \quad \text{and} \quad \mathbb{I}_d \preceq \Sigma \preceq \kappa \mathbb{I}_d,\]
even though, crucially, \emph{the samples $\vec X_i$ themselves are unbounded}. Moreover, both upper bounds ($R$ and $\kappa$) are mild and well-motivated: even if \textsc{LearnGaussian-hd} is applied to a sequence of datasets where these grow sub-exponentially, the sample complexity remains polynomial. 
Additionally, notice that $\mathbb{I}_d\preceq\Sigma$ comes w.l.o.g.: as long as the smallest eigenvalue of $\Sigma$ is non-zero, we can rescale the vectors $\vec X_i$ to ensure that this holds. If an eigenvalue of $\Sigma$ is zero, then the distribution is degenerate: we can then apply \textsc{LearnGaussian-hd}  in the subspace spanned by the features (in which $\Sigma$ will have full rank).

\cite{gautamHighDimensional} efficiently learn a symmetric matrix $A$, termed the \emph{preconditioner} of the Gaussian distribution,  that satisfies
$\mathbb{I}_d \preceq A \Sigma A \preceq O(1)\mathbb{I}_d$. Multiplying the input samples with this preconditioner thus makes the Gaussian inputs nearly spherical, which reduces the geometry to the one-dimensional setting, previously studied by \cite{vadhanConfidenceIntervals}.

\section{Problem Formulation}
\label{section:problem-formulation}
In this section, we formally define the regression settings we are interested in, namely, the Least Squares Fitting, the Binary Regression and the (standard) Linear Regression problems, as well as the associated technical assumptions we make.\\

\noindent\textbf{Least Squares Fitting.} In the Least Squares Fitting problem, we observe labeled examples $(\vec X_i, y_i) \in \reals^d \times \reals$, and wish to produce an $(\epsilon,\delta)$-differentially private version of the Least Squares Estimator (LSE):
\begin{align}
\label{eq:lse-estim}
    \vec \beta^\star &= \argmin_{\vec \beta \in \reals^d} \sum_{i=1}^{n} \left( y_i - \vec \beta^T \vec X_i \right)^2 \\  
 \label{eq:lse-estim2}
 &= \left( \frac{1}{n} \sum_{i=1}^n \vec X_i \vec X_i^T \right)^{-1} \left( \frac{1}{n} \sum_{i=1}^n y_i \vec X_i \right) \\
 \label{eq:lse-estim3}
 &= \left( \frac{1}{n} X^T X \right)^{-1} \frac{1}{n} X^T \vec y\, ,
\end{align}
where $X=[\vec{X}_i]_{i=1}^n\in\reals^{n\times d}$ is the matrix with feature vectors as rows and $\vec y =[y_i]_{i=1}^n\in\reals^d$ is the vector of labels, respectively.
In contrast to the Binary and Linear Regression problems below, we make no prior assumption on how labels $y_i$ are linked to features $\vec {X}_i$; crucially, our differentially private algorithm \emph{must not rely} on any presumed boundedness of features $\vec X_i$.
We make the following technical assumption:

\begin{assumption}
\label{asm:gauss1}
Labeled examples $(\vec X_i, y_i)$, $i=1,\ldots,n$, are i.i.d. Moreover, $\vec X_i \in \reals^d$ are sampled from a Gaussian distribution $\mathcal{N}(\vec \mu, \Sigma)$ 
satisfying the following conditions:
\begin{align}\label{eq:parambounds}
\left\| \vec{\mu} \right\|_2 \leq R 
~~~\text{and}~~~ \mathbb{I}_d \preceq \Sigma \preceq \kappa \mathbb{I}_d
\,,
\end{align}
while the labels satisfy
$\frac{1}{\rho} \leq |y_i| \leq c$
for some universal parameters $\rho, c, \kappa, R > 0$.
\end{assumption}
The  assumptions in \Cref{eq:parambounds} are also made by \cite{gautamHighDimensional} in the context of Gaussian estimation. As discussed in \Cref{sec:paremestimation}, both upper bounds are natural, while the lower bound on the covariance comes without any loss of generality. Crucially, in contrast to the majority of prior works on regression, samples $\vec{X}_i$ are indeed unbounded, as they are sampled from $\mathcal{N}(\vec \mu, \Sigma)$. Finally, the boundedness of the outputs $y_i$, $i\in [n]$, is a requirement we share with other works (e.g., \cite{alabi2020arxiv, wang, kifer12, ERM_Zhang}), and clearly applies to, e.g., binary  classification; we also study unbounded labels in the Linear Regression setting.\\

\noindent\textbf{Binary Regression.} In the Binary Regression setting, we additionally assume that the labels $y_i$ are binary (i.e., $y_i\in \{-1,+1\}$), and are produced by a Generalized Linear Model (GLM) linking these binary labels to features. In contrast to the previous setting, this GLM is parameterized by a ``true'' $\vec{\beta}\in\reals^d$ (see \Cref{asm:glm} below). Our goal is to give an estimate of this $\vec{\beta}$ again via \emph{the same} $(\epsilon,\delta)$-differentially private version of the LSE given by \Cref{eq:lse-estim}.
In particular, \emph{in addition} to \Cref{asm:gauss1}, we make the following assumption in the Binary Regression setting:
\begin{assumption}
\label{asm:glm}
There exists a $\vec \beta\in \reals^d$ such that, given $\vec X_i\in \reals^d$ and
for all $i \in [n]$,
\begin{align}\label{eq:glm}
\Pr[y_i = +1 | \vec X_i] = f(\vec \beta^T \vec X_i),
\end{align}
where $f : \reals \rightarrow [0,1]$ is a non-decreasing, continuously differentiable function satisfying $\lim_{x \rightarrow -\infty} f(x) = 0$ and $\lim_{x \rightarrow \infty} f(x) = 1$.
Moreover, the features $\vec X_i$ are zero-mean, i.e., $\vec{\mu}=\E[\vec X_i] = \vec 0$.
\end{assumption}

The probabilistic model defined by 
\Cref{eq:glm}
holds for many important practical settings. For instance, it holds for logistic regression, where the link function is $f(x) = 1/\left( 1+e^{-x} \right)$. It also holds for Support Vector Machines (SVMs) with linearly separable data. We discuss this in more detail in \Cref{appendix:discussion-glm-svm-log}.

Finally, our assumption that $\vec \mu = \vec 0$ is common (see, e.g., \cite{bayesian_glms, thecostofprivacy_GLMs_cai2020, daskalakis_neurips2020_gaussian, BernsteinSheldon2019, sheffet_DP_OLS_17, muratErdogdu}) and well-motivated in the context of our Binary Regression setting: even ignoring privacy considerations, the sample complexity guarantees of any estimator will degrade rapidly as $\vec{\mu}$ gets farther away from the origin. This is precisely because, under Gaussian covariates, the fraction of samples of one class will decrease exponentially as the distance of $\vec{\mu}$ from the separating hyperplane (that passes through the origin) increases.

\paragraph{Linear Regression.} A natural question is whether we can extend our guarantees beyond bounded labels. To this end, we finally consider the standard Linear Regression setting (with Gaussian errors):
\begin{assumption}
\label{asm:standard_linear_regression}
Labeled examples $(\vec X_i, y_i)$, $i=1,\ldots,n$, are i.i.d., where $\vec X_i \in \reals^d$ are sampled from the Gaussian distribution $\calN(\vec \mu, \Sigma)$ 
satisfying $\mathbb{I}_d \preceq \Sigma \preceq \kappa \mathbb{I}_d$
for some universal parameter $\kappa > 0$.
Moreover, there exists a $\vec \beta\in \reals^d$ and a $\sigma_\epsilon>0$ such that, given $\vec X_i\in \reals^d$,
\begin{align}\label{eq:gaussian_error_linear_model}
y_i = \vec\beta^T \vec X_i + \eps_i,\quad \text{for all}~i=1,\ldots,n
\, ,
\end{align}
where $\eps_i$ are i.i.d.~samples from $\calN(0, \sigma_\eps^2)$.
\end{assumption}
Note that, in this setting, labels $y_i$ are themselves Gaussian and, therefore, unbounded. Our goal here is again to produce a differentially private estimate for the ``ground truth'' vector $\vec{\beta}$.

\section{Main Results}
\label{section:main-results}
We  formally state our results in this section. Our theorems provide $(\frac{\eps^2}{2} + \eps\sqrt{2\log(1/\delta)}, \delta)$-DP guarantees for the Least Squares Fitting, the Binary and Linear Regression settings. This guarantee is, in essence, equivalent to $(\eps, \delta)$-DP. For a more detailed discussion on this issue, we refer the reader to \Cref{section:equiv-priv}. We focus here on the statement of our main results and conclusions drawn from them; an overview of the technical challenges we face when proving these results and the novel techniques we employ to address them  can be found in \Cref{section:technical-overview}.

\subsection{Least Squares Fitting}

Our differentially private LSE for the Least Squares Fitting setting is summarized in \Cref{algo:betahat}. In short, we compute DP estimates of the quantities
\begin{align*}(X^T X/n)^{-1} \quad \text{and}\quad X^T \vec y/n,\end{align*}
whose product, by \Cref{eq:lse-estim3}, yields the LSE $\vec \beta^
\star$.

The estimation of the first quantity proceeds as follows. Having access to the $n$ i.i.d. samples $(\vec X_i, y_i) \in \reals^d \times \reals$, where $\vec X_i \sim \calN(\vec \mu, \Sigma), i \in [n]$, \Cref{algo:betahat} initially privately
computes differentially private estimates $(\wh{\vec \mu}_{\vec X}, \wh{\Sigma}_{\vec X})$
of the mean and covariance matrix of the $d$-dimensional Gaussian distribution $\calN(\vec \mu, \Sigma)$, using the algorithm $\textsc{LearnGaussian-hd}$, discussed in \Cref{sec:paremestimation}. These estimates, that satisfy the guarantees indicated in \Cref{theorem:gautam-gaussian},  can be used to estimate  $(X^T X/n)^{-1}$ via the relationship:
\[
\frac{1}{n}\sum_{i=1}^n \vec{X}_i \vec X_i^T
\approx
\E[\vec X_i \vec X_i^T]
\approx
\wh{\Sigma}_{\vec X} + \wh{\vec \mu}_{\vec X} \wh{\vec \mu}_{\vec X}^T
\, .
\]
The second quantity, i.e., the term $X^T \vec y/n$, is somewhat harder to estimate in a differentially private fashion, as constituent terms $y_i \vec X_i$ are \emph{not} Gaussian. The boundedness of variables $y_i$, however, ensures that these terms are sub-gaussian.
As an important technical contribution, we design differentially private algorithms that operate in the sub-gaussian regime (see \textsc{LearnSubGaussian-hd} in \Cref{appendix:variants-subgaussian}),
extending the analysis of \cite{gautamHighDimensional} and \cite{vadhanConfidenceIntervals},
and obtain a private mean estimate $\wh{\vec \mu}_{\vec X, y}$ for the sub-gaussian random vectors $y_i \vec X_i$. 

Armed with these  estimates, our differentially private LSE 
 is finally given by:
\begin{align}
\label{eq:dpestimate}
\wh{\vec\beta} = \left( \wh{\Sigma}_{\vec X} + \wh{\vec\mu}_{\vec X} \wh{\vec\mu}^T_{\vec X} \right)^{-1} \wh{\vec\mu}_{\vec X, y}
\, ,
\end{align}
whose privacy follows from appropriate composition rules.
We refer to the resulting algorithm, summarized in  \Cref{algo:betahat}, as \textsc{PrivLearnLSE}.
Our main result with respect to the privacy and accuracy of this estimator is as follows:
\begin{theorem}
[Privacy and Accuracy of $\wh{\vec\beta}$ in Private Least Squares Fitting]
\label{theorem:main-result1}
Under \Cref{asm:gauss1} with parameters $(\kappa, c, \rho, R)$, for all privacy parameters $\eps, \delta > 0$, accuracy parameters $\alpha, \eta > 0$ and confidence $\gamma \in (0,1)$, \textsc{PrivLearnLSE} (defined in \Cref{algo:betahat}) is $(\frac{\eps^2}{2} + \eps\sqrt{2\log(1/\delta)}, \delta)$-differentially private.
Moreover, if the number of labeled examples is at least: 
\begin{align*}
n &= \wt{O} \left( 
\frac{d^{3/2} \sqrt{\log(\kappa \rho c)} \cdot  \polylog \left(\frac{1}{ \gamma \delta}\right)}{\eta^2 \eps} \right) \\
& + (1+R) \cdot \wt{O} \left( \frac{d^{3/2} \sqrt{\log \kappa} \cdot \polylog \left(\frac{1}{ \gamma \delta} \right) }{\alpha^2 \eps} \right) \,,
\end{align*}
then, 
\textsc{PrivLearnLSE} 
runs in $\poly(n)$ time
and,
with probability at least $1-O(\gamma)$, successfully returns an estimate $\wh{\vec\beta} \in \reals^d$ that satisfies:
\[
\left\| \wh{\vec\beta} - \vec\beta^\star \right\|_2^2
\leq O\left( \alpha^2 \right)
\cdot \left\| \Sigma^{1/2} \vec\beta^{\star} \right\|_2^2
+ O\left( \eta^2 \right) \cdot c^2
\,,
\]
 with respect to the LSE $\vec\beta^*$.
\end{theorem}
We have provided a simplified bound in the number of samples; the precise sample complexity  and the theorem's proof can be found in \Cref{appendix:proof-main-result1}.
For a proof sketch, we refer to \Cref{sec:sketch-1}.
Intuitively, 
the number of samples we require grows as $\wt{O}(d^{3/2})$, slightly more favorably than the covariance estimation case of \citet{gautamHighDimensional}. Moreover, the number of samples again grows polylogarithmically on  $\kappa$ (the bound on the covariance spectral norm) but linearly (rather than polylogarithmically) on $R$, the bound on the mean.

The \textsc{LearnGaussian-HD} routine (see also \textbf{Line 6} of \Cref{algo:betahat}) is the algorithm of \cite{gautamHighDimensional} (as discussed in \Cref{sec:paremestimation} and  \Cref{subsubsection:algo_variant_overview}). The original algorithm by \cite{gautamHighDimensional} requires knowledge of both upper bounds $\kappa$ and $R$, but by switching the privacy guarantee from zero-concentrated DP to  $(\epsilon,\delta)$-DP, we remove the requirement of  prior knowledge of $R$, even though we still require $\kappa$ as input.
This adaptation can be found in \Cref{section:equiv-priv}. In contrast, our routine \textsc{LearnSubGaussian-HD} (see \textbf{Line 8} of \Cref{algo:betahat}), described in \Cref{subsubsection:algo_variant_overview}, departs from the one of \cite{gautamHighDimensional}, the main difference being that it operates (and comes with guarantees for) sub-gaussian vectors. For further details about the  modifications required to accomplish this, we refer the reader to \Cref{appendix:variants-subgaussian}.

\begin{algorithm}[!t] 
\caption{Differentially Private LSE.}
\label{algo:betahat}
\begin{algorithmic}[1]
\State \textbf{Input:} $(X, \vec y) = (\vec X_i, y_i)_{i \in [n]}$ with $\vec X_i \sim \calN(\vec \mu, \Sigma)$, where $\vec \mu, \Sigma$ are unknown and $n$ satisfies \Cref{theorem:main-result1}.
\State \textbf{Parameters:} Privacy $\eps, \delta > 0$, accuracy $\alpha, \eta > 0$, confidence $\gamma \in (0,1)$,
covariance spectral norm bound $\kappa$, upper bound of labels $c$. 
\State \textbf{Output:} Estimate $\wh{\vec\beta}$ that approaches the LSE $\vec\beta^\star$ in $L_2$ norm with high probability.
\vspace{1mm}
\Procedure{PrivLearnLSE}{$(X, \vec y), \eps, \delta, \alpha, \eta, \gamma, \kappa$} 
\State $L \gets \{ \Theta(\eps), \Theta(\delta), \Theta(\alpha), \gamma, \kappa \}$
\State $(\wh{\vec \mu}_{\vec X}, \wh{\Sigma}_{\vec X}) \gets \textsc{LearnGaussian-hd}(\{\vec X_i\}_{i
\in [n]}, L)$
\State $L \gets \{\Theta(\eps), \Theta(\delta), \Theta(\eta), \gamma, c^2 \kappa\}$
\State $\wh{\vec \mu}_{\vec X, y} \gets \textsc{LearnSubGaussian-hd}(\{y_i \vec X_i\}_{i \in [n]}, L)$
\State $\mathrm{M} \gets \wh{\Sigma}_{\vec X} + \wh{\vec\mu}_{\vec X} \wh{\vec\mu}^T_{\vec X}$
\State \textbf{if} $\mathrm{M}$
is not invertible\footnotemark\ 
\textbf{then} Output $\perp$
\State Output the private estimate $\wh{\vec\beta} = 
\mathrm{M}^{-1} \wh{\vec\mu}_{\vec X, y}$
\EndProcedure
\end{algorithmic}
\end{algorithm}
\footnotetext{The invertibility of the matrix in \textbf{Line 10} holds with high probability; we account for the bad non-invertibility event in the $O(\gamma)$ failure probability of \Cref{theorem:main-result1}.}

\subsection{Binary Regression}

We next turn our attention to the Binary Regression setting, in which \emph{both} \Cref{asm:gauss1} and \Cref{asm:glm} apply. We study the properties of \textsc{PrivLearnLSE} (\Cref{algo:betahat}) under these assumptions; the only (slight) modification of \Cref{algo:betahat}, compared to the previous setting, is that we no longer need the estimate $\wh{\vec \mu}_{\vec X}$, as \Cref{asm:glm} states that $\vec \mu = \vec 0$. Hence, we set  $\wh{\vec \mu}_{\vec X} = \vec 0$ in \Cref{eq:dpestimate}, with the remaining terms computed as in the previous section. We show that the resulting algorithm has the following guarantees:

\begin{theorem}
[Privacy and Accuracy of $\wh{\vec\beta}$ in Private Binary Regression]
\label{theorem:main-result2}
Under \Cref{asm:gauss1} with covariance parameter $\kappa$ and \Cref{asm:glm} with true parameter $\vec \beta \in \reals^d$, for every privacy parameters $\eps, \delta > 0$, accuracy parameters $\alpha, \eta > 0$ and confidence $\gamma \in (0,1)$,  \textsc{PrivLearnLSE} (defined in \Cref{algo:betahat}) with $\wh{\vec \mu}_{\vec X} = \vec 0$  is  $(\frac{\eps^2}{2} + \eps\sqrt{2\log(1/\delta)}, \delta)$-differentially private. Moreover, if the number of labeled examples is at least:
\begin{align*}
n &= \wt{O} \left( 
\frac{d^{3/2} \cdot \polylog \left( \frac{1}{ \gamma \delta} \right) }{\eps} \cdot \max \left\{ \frac{\sqrt{\log \kappa}}{\eta^2}, \frac{1}{\alpha^2} \right\}
\right)
\, ,
\end{align*}
then
\textsc{PrivLearnLSE}
runs in $\poly(n)$ time and,
with probability at least $1-O(\gamma)$, successfully returns an output estimate $\wh{\vec\beta} \in \reals^d$ that satisfies
\begin{align}
\label{eq:k}
\|\wh{\vec\beta}-k\vec\beta\|_2^2
\leq O( \alpha^2 )
\Big( 1 \!+\! \big\|k \Sigma^{1/2} \vec\beta \big\|_2^2 \Big)
\!+\! O( \eta^2)
\,,\!\!
\end{align}
where $k =
\frac{2n}{n-d-1}
\E
\left[ f'\left( \vec\beta^T \vec X_i \right) \right]
\, .$
\end{theorem}
As before, we have provided a simplified version of the exact number of samples.
The exact expression and the proof of the theorem are in \Cref{appendix:proof-main-result2}. The proof sketch can be found at \Cref{sec:sketch-2}. As in \Cref{theorem:main-result1}, the sample complexity grows as $d^{3/2}$, and is merely polylogarithmic on $\kappa$.
Moreover, as in classic (non-DP) work on binary regression via LSE \citep{kadioglu2021sample,muratErdogdu, sun2014learning, brillinger2012generalized}, our estimator learns the underlying ``true'' $\vec{\beta}$ up to a scaling factor $k$, that depends on the ``sharpness'' of the model function $f$ (via its derivative $f'$). We note that, to discover the hyperplane separating positive from negative labels, it indeed suffices to learn only the direction of $\vec\beta$, not its magnitude, since a separating hyperplane is fully defined  by this direction. 

To further elaborate  on the effect of $k$: by \Cref{asm:glm}, $\vec\beta^T\vec{X}_i$ is a zero mean Gaussian, while $f'$ tends to zero as its argument reaches either $+\infty$ or $-\infty$. Hence, the expectation that determines $k$ very much depends by the behavior of  $f'$ around 0. That is, if $f$ is relatively flat (i.e., binary labels are ``noisy''), $k$ will be small, and more samples will be needed to achieve a better numerical accuracy in \Cref{eq:k}; the converse is true when $f$ is ``sharp'' (e.g., a sigmoid close to the sign function), and labels are less noisy. This dependence of the estimate accuracy on the noise inherent in the GLM (via the model function $f$) is natural.

\subsection{Linear Regression}
In this model, the labels $y_i$ are assumed to be generated from an underlying ``true'' linear model $\vec\beta^T \vec X_i$ (with a Gaussian error), thus being unbounded, for some regression coefficient $\vec\beta \in\reals^d$. Our goal is to estimate this ``true'' underlying $\vec\beta$ in a differentially private way.
We provide the following algorithm for this task.
For each drawn labeled example $(\vec X, y)$, the algorithm creates the  vector $\vec Z = (\vec X, y)^T\in \reals^{(d+1)}$. Observe that this random vector is also Gaussian with a covariance matrix $\Sigma'\in \reals^{(d+1)\times (d+1)}$, given by:
\begin{align}
\label{eq:block_matrix-main}
\Sigma' = 
\begin{bmatrix}
\Sigma & \Sigma\vec\beta \\
\vec\beta^T \Sigma & \sigma_\eps^2 + \vec\beta^T \Sigma \vec\beta
\end{bmatrix}
\,,
\end{align}
where $\vec \mu, \Sigma, \sigma_\eps^2$ are the parameters of \Cref{asm:standard_linear_regression}. The algorithm, then, proceeds as follows. First, it computes a differentially private estimate $\wh{\Sigma}$ of $\Sigma$ using $n$ samples of $\vec X_i$ via  the routine $\textsc{LearnGaussian-hd}$, discussed in \Cref{sec:paremestimation}. Then, using $n$ additional samples $\vec Z_i = (\vec X_i, y_i)^T$, it computes a differentially private estimate $\wh{\Sigma'}$, again via   $\textsc{LearnGaussian-hd}$.
From \Cref{eq:block_matrix-main}, the first $d$ elements of the last column of $\wh{\Sigma'}$ can be used as a DP estimate    $\wh{\Sigma\vec\beta}$ of  $\Sigma\vec\beta$.\footnote{Note that  the first $d$ columns and rows of $\wh{\Sigma'}$ can also be used as a DP estimate $\wh{\Sigma'}$ of $\Sigma$; we nevertheless estimate this separately, to ensure the statistical independence of the two estimates.}   
Finally, the algorithm uses these two estimates to output: 
\begin{align*}
\wh{\vec\beta} = \wh{\Sigma}^{-1} \wh{\Sigma\vec\beta}
\,.
\end{align*}
A formal description of this algorithm can be found in  \Cref{algo:linear-regression} in \Cref{app:linear}. Our result with respect to its privacy and accuracy is as follows:
\begin{theorem}
[Privacy and Accuracy of $\wh{\vec\beta}$ in Private Linear Regression]
\label{theorem:main-result3}
Under \Cref{asm:standard_linear_regression} with parameter $\kappa$ and true vector $\vec \beta \in \reals^d$, for all privacy parameters $\eps, \delta > 0$, accuracy parameters $\alpha, \eta > 0$ and confidence $\gamma \in (0,1)$,
there exists an algorithm (see \Cref{algo:linear-regression}) that is
$(\frac{\eps^2}{2} + \eps\sqrt{2\log(1/\delta)}, \delta)$-differentially private, and if the number of samples is at least: 
\begin{align*}
n &= \wt{O} 
\left( 
\frac{
d^{3/2} \cdot \polylog \left( \frac{1}{\gamma \delta } \right) }
{ \eps} 
\max 
\left \{
\frac{\sqrt{\log (\kappa(\Sigma'))}}{\eta^2}, \frac{1}{\alpha^2} 
\right \}
\right)
\, ,
\end{align*}
then, it runs in $\poly(n)$ time and, with probability at least $1-O(\gamma)$, the output estimate $\wh{\vec\beta} \in \reals^d$ and the ``true'' regression coefficient $\vec\beta$ satisfy:
\begin{align*}
\left\| \wh{\vec\beta} - \vec\beta \right\|_2^2
\leq O\left( \alpha^2 \right)
\cdot \left\| \Sigma^{1/2} \vec \beta\right\|_2^2
+ O\left( \eta^2 \right) \cdot \lambda_\text{max}^2(\Sigma')
\, ,
\end{align*}
where $\kappa(\Sigma') = \frac{\lambda_\text{max}(\Sigma')}{\lambda_\text{min}(\Sigma')}$ is the condition number of the block matrix $\Sigma'$ as in \Cref{eq:block_matrix-main}.
\end{theorem}

The exact sample complexity bound and the theorem's proof can be found in \Cref{app:linear}. A short proof sketch is provided in \Cref{sec:sketch-3}. As in our previous results, the sample complexity scales as $d^{3/2}$; also, it is polylogarithmic on the condition number of $\Sigma'$.\footnote{
Again,  the  non-invertibility of the matrix $\wh{\Sigma}$ is a low probability event and is fully captured by the $O(\gamma)$ probability of failure, as also indicated in our proof.}

\section{Technical Overview}
\label{section:technical-overview}
In this section, we provide a sketch of our technical contributions with respect to the proofs of Theorems \ref{theorem:main-result1},  \ref{theorem:main-result2}, and \ref{theorem:main-result3}.

\subsection{\Cref{theorem:main-result1}: Proof Sketch}
\label{sec:sketch-1}
We begin with \Cref{theorem:main-result1},  which deals with the Least Squares Fitting problem. Our goal is to privatize the Least Squares Estimator (see \Cref{eq:lse-estim}) without significant accuracy loss. Hence, the differentially private algorithm (see \Cref{algo:betahat})
computes a quantity $\wh{\vec \beta}$ that is asymptotically the same as the Least Squares Estimate of \Cref{eq:lse-estim}:
\begin{align}
\label{eq:lse-estim-tech-overview}
\vec\beta^\star = \left( \frac{1}{n} \sum_{i=1}^n \vec X_i \vec X_i^T \right)^{-1} \left( \frac{1}{n} \sum_{i=1}^n y_i \vec X_i \right)
\, .
\end{align}

The structure of this estimate (product of two terms) suggests privatizing each term separately, thereby motivating \Cref{algo:betahat}. To ensure that the desired privacy property holds, the key idea is to apply the composition of differentially private mechanisms (see \Cref{def:dp-composition}), hence affording privacy to the whole algorithm. It thus suffices to consider privatized estimates of the individual terms.

The key conceptual observation for our main result is that the second term in \Cref{eq:lse-estim-tech-overview} consists, in fact, of sub-gaussian vectors. At a technical level, we have to expand the mean and covariance estimation procedures for Gaussian distributions to the sub-gaussian regime. More to that, in order to reduce as much as possible the dependence on the range of the mean value $R$ of the feature vectors $\vec X_i$, we modify the multivariate mean estimation analysis of \cite{gautamHighDimensional} to hold for unbounded mean feature vectors. As a technical tool, we use an alternative guarantee (see \Cref{lemma:vadhan_variation} in \Cref{appendix:variants-subgaussian}) on mean estimation which allows us to disengage the concentration bounds from the bound on the mean, in the case of $(\epsilon,\delta)$-DP.

Even using those variants of the algorithms, we still have to satisfy a stronger privacy desideratum.
In particular, \Cref{theorem:main-result1} requires privacy guarantees for pairs $(\vec X_i, y_i)$. 
However, \textbf{Line~8} of \Cref{algo:betahat} affords privacy guarantees for the entire sub-gaussian terms $y_i \vec X_i$.
So, it is not straightforward how to achieve the more general privacy guarantee of altering the individual $(\vec X_i, y_i)$ pairs.
In \Cref{subsection:priv_guarant}, we
establish the desired privacy guarantee for $(\vec X_i, y_i)$. 

For the desired accuracy guarantee on \Cref{algo:betahat}, we have to control the quantity $\left\| \wh{\vec\beta} - \vec\beta^\star \right\|_2^2$ (see \Cref{app:acc}). At a first sight, the above expression cannot be handled by standard concentration of measure phenomena. However, we provide a non-trivial decomposition:
\[
\wh{\vec\beta} - \vec\beta^\star = \left(\wh{\Sigma} + \wh{\vec\mu}_{\vec X} \wh{\vec\mu}_{\vec X}^T \right)^{-1} \left( - Q_1 \vec\beta^\star + \vec Q_2 \right)
\, ,
\]
using the below quantities that we introduce:
\[
Q_1 = \wh{\Sigma} + \wh{\vec\mu}_{\vec X} \wh{\vec\mu}_{\vec X}^T - \frac{1}{n} X^T X
,~~~\text{and}~~~
\vec Q_2 = \wh{\vec\mu}_{\vec X, y} - \frac{1}{n} X^T \vec y
\, ,
\]
where $\wh{\Sigma}, \wh{\vec\mu}_{\vec X}, \wh{\vec\mu}_{\vec X, y}$ are the private outputs of the algorithms described in \Cref{algo:betahat}, and $X, \vec y$ are the design matrix and the labels vector, respectively. This decomposition, when altered in geometry for normalization purposes by a transformation $\vec w = \Sigma^{1/2} \vec\beta$ and $\wh{\vec w} = \Sigma^{1/2} \wh{\vec\beta}$, enables us to control each term individually and obtain the desired bounds. The intuition behind this decomposition lies in the fact that both $Q_1$ and $\vec Q_2$ vanish asymptotically (and so $\wh{\vec\beta}$ tends to $\vec\beta^\star$), as the number of samples $n$ increases.

The bounds on $Q_1, \vec Q_2$ are handled by further decomposing into the difference of private quantities and their actual values ($\Sigma, \vec\mu_{\vec X}, \vec\mu_{\vec X, y}$) and between empirical quantities and the actual values. To obtain tighter bounds on the individual terms of difference of private quantities and actual values, we use the private preconditioner matrix in our analysis, which allows us to avoid a strict dependence on the largest eigenvalue $\kappa$ of the covariance matrix $\Sigma$ in our bounds (see \Cref{theorem:main-result1}). For a detailed proof of \Cref{theorem:main-result1}, see \Cref{appendix:proof-main-result1}.

\subsection{\Cref{theorem:main-result2}: Proof Sketch}
\label{sec:sketch-2}
As far as our second main result (\Cref{theorem:main-result2}) is concerned, the key conceptual contribution is to introduce a new estimator $\vec\beta_s^\star$ (solely for the purposes of the analysis) that is defined with the help of $n$ additional samples $(\vec X_i, y_i)$ (for a total of $2n$ samples) as follows:
\[
\vec \beta_s^\star =
\left( \frac{1}{n} \sum_{i=n+1}^{2n} \vec X_i \vec X_i^T \right)^{-1} \left( \frac{1}{n} \sum_{i=1}^n y_i \vec X_i \right)
\, .
\]

This estimate resembles the Least Squares Estimate $\vec\beta^\star$ but \emph{crucially introduces independence} between the two terms that constitute the Least Squares Estimate. This independence of the two terms is pivotal for proving that the estimate $\vec\beta_s^\star$ is an unbiased up to a multiplicative factor estimate of the true regression coefficient $\vec\beta$. In turn, this crucial observation is used to prove that our private estimate $\wh{\vec\beta}$ (see \Cref{algo:betahat}) is close to the true regression coefficient $\vec\beta$ up to a multiplicative factor, since the proof of \Cref{theorem:main-result1} holds  even for the Least-Squares-resembling estimate $\vec\beta_s^\star$ (because of the independent handling of the aforementioned quantities $Q_1, \vec Q_2$). At a technical level, the above discussion is a result of probabilistic tools, such as the high-dimensional geometry of Wishart matrices \citep{anderson_multivariate_statistics}. For a detailed proof of \Cref{theorem:main-result2}, see \Cref{appendix:proof-main-result2}.

\subsection{\Cref{theorem:main-result3}: Proof Sketch}
\label{sec:sketch-3}
Finally, we briefly discuss the techniques behind \Cref{theorem:main-result3}. Recall that for the standard Linear Regression problem with true vector $\vec \beta$, our algorithm outputs the private estimate
$\wh{\vec\beta} = \wh{\Sigma}^{-1} \wh{\Sigma\vec\beta}$, as mentioned after \Cref{eq:block_matrix-main}. On  one hand, the privacy guarantee follows from the composition theorems. On the other hand, for the accuracy guarantee, we have to control the quantity $\| \wh{ \vec \beta} - \vec \beta \|_2^2$. The main technical challenge for this step it to provide tight bounds for the eigenvalues of the block matrix $\Sigma'$ of \Cref{eq:block_matrix-main}. 
In particular, we have to draw sufficiently many samples in order to control the quantities $\| \Sigma^{1/2} \wh{\Sigma}^{-1} \Sigma^{1/2} \|_2^2$ and $\| \Sigma^{-1/2}(\wh{\Sigma \vec \beta} - \Sigma \vec \beta) \|_2^2$ dealing with our estimates $\wh{\Sigma}$ and $\wh{\Sigma \vec \beta}$. The first quantity is a constant, given roughly $n = \Omega(d^{3/2}\sqrt{\log \kappa}/\eps)$ samples, using properties of the \textsc{LearnGaussian-HD} algorithm and concentration of random matrices. The second quantity is more challenging and is controlled by the maximum eigenvalue of $\Sigma'$, with high probability, after roughly
$n = \Omega(d^{3/2} \sqrt{\log(\kappa(\Sigma'))}/\eps)$ samples are drawn, where $\kappa(\Sigma')$ is the condition number of the block matrix of \Cref{eq:block_matrix-main}. To upper bound the condition number, we exploit bounds for eigenvalues of block matrices \citep{lower_upper_bounds_block_matrix}, and show that in our setting, these are tight for $\kappa(\Sigma')$ (see \Cref{subsec:proof_of_last_column_norm}).

\section{Conclusion}
\label{section:conclusion}

We provide and analyze estimators for inference in three regression settings with unbounded covariates, formally proving that they are private and efficient. We believe that the line of work on unbounded covariates is of great interest with respect to both theory and practice. Potential future research based on this work includes, for instance, relaxing the i.i.d.~assumptions on the provided data (to account for potential dependencies among feature vectors). In addition, lower bounds in differentially private regression regimes are either elusive or sub-optimal (see, e.g., \cite{wang}); examining possible lower bounds in unbounded regimes for regression-like environments is another promising future direction.

\subsection{Limitations}

For the above analysis, we have considered the case of Gaussian marginals and have extended recent differentially private techniques on mean and covariance estimation \citep{gautamHighDimensional, vadhanConfidenceIntervals} to the sub-gaussian regime.
For the detailed hypotheses upon which the aforementioned procedures were provided, the reader is encouraged to review \Cref{section:problem-formulation}, where all of the relevant assumptions are clearly indicated.

The focus of this work is in its nature theoretical. Supplementally to the theory, we believe that the community would benefit from additional experimental studies of the proposed methods.
In fact, the design of practical algorithms is a strand of research of significant independent interest, since practical applications are able to immensely benefit from unbounded estimation procedures: see, e.g., the work of \cite{gautam_practical} that considers practical differentially private Gaussian mean and covariance estimation procedures.
Thus, we believe that the practical extension of our results and relevant experiments are a natural and interesting premise for future work.

\subsubsection*{Acknowledgements}
We thank the anonymous reviewers for useful remarks and comments on the presentation of our manuscript.
The most significant part of this work was performed while Jason Milionis was an undergraduate student at the National Technical University of Athens. This work was partially supported by a research fellowship from the Costis M. Lemos Foundation.
Dimitris Fotakis and Alkis Kalavasis were supported by the Hellenic Foundation for Research and Innovation (H.F.R.I.) under the ``First Call for H.F.R.I. Research Projects to support Faculty members and Researchers and the procurement of high-cost research equipment grant,'' project BALSAM, HFRI-FM17-1424.
Stratis Ioannidis was supported by the National Science Foundation (through grants 2112471, 2107062, and 1750539) and by the Niarchos Foundation, through the  Greek Diaspora Fellowship Program.

\bibliography{references}


\clearpage
\appendix

\thispagestyle{empty}


\section{Additional Preliminaries}
\label{appendix:add-prelim}
In this section, we provide some tools required for our analysis.\\

\noindent\textbf{Sub-gaussianity Tools.} We begin with definitions of the sub-gaussian norm for univariate and multivariate random variables, and move on to some of their properties \citep{hdp_book} that we later use.

\begin{definition}
[Sub-gaussian random variable]
A random variable $X$ is called a sub-gaussian random variable if there exists $K > 0$ such that, for all $\lambda : |\lambda| \leq 1/K$,
\[
\E\left[ \exp\left( \lambda^2 X^2 \right) \right]
\leq
\exp(\lambda^2 K^2)
\, .
\]
The smallest $K$ for which the above property holds is called the sub-gaussian norm of $X$, and is denoted as $\left\|X\right\|_{\psi_2}$.
\end{definition}
\begin{definition}
[Sub-gaussian random vector]
\label{def:subg}
A random vector $\vec X\in\reals^d$ is called a sub-gaussian random vector if for all $\vec u\in\reals^d$, the inner product $\langle \vec X, \vec u \rangle$ is a sub-gaussian random variable. The sub-gaussian norm of a sub-gaussian random vector is defined as follows:
\[
\left\| \vec X \right\|_{\psi_2}
= \sup_{\vec u\in S^{d-1}}
\left\| \langle \vec X, \vec u \rangle \right\|_{\psi_2}
\, ,
\]
where $S^{d-1} = \{ \vec u\in\reals^d : \|\vec u\|_2=1 \}$ is the $d$-dimensional unit sphere.
\end{definition}

\begin{lemma}
[Properties of the sub-gaussian norm]
\label{lemma:subg_properties}
Let $\vec X$ be a sub-gaussian random vector. Then, the following hold:
\begin{itemize}
    \item For every constant $c>0$, $c \vec X$ is a sub-gaussian random vector, with $\left\| c\vec X \right\|_{\psi_2} = c \left\| \vec X \right\|_{\psi_2}$.
    \item If $\E[\vec X] = \vec\mu_{\vec X}$, then $\vec X - \vec\mu_{\vec X}$ is a sub-gaussian random vector, with \[\left\| \vec X - \vec\mu_{\vec X} \right\|_{\psi_2} \leq C \left\| \vec X \right\|_{\psi_2}\,,\] for a universal constant $C>0$.
\end{itemize}
\end{lemma}

\paragraph{Differential Privacy Tools.} We continue with a slightly different definition of differential privacy that is roughly equivalent with the classical definition, according to \Cref{lemma:zCDP_equivalence} \citep{zcdp_steinke}.

\begin{definition}
[Zero-concentrated DP (zCDP)]
\label{def:zCDP}
A randomized mechanism $M : \calX^n \to \calY$ satisfies $\rho$-zCDP if for every pair of neighboring datasets $X, X' \in \calX^n$ that differ on at most one element, and for every $\alpha \geq 1$,
\[
D_\alpha \left( M(X) || M(X') \right)
\leq \rho\alpha
\, ,
\]
where $D_\alpha(P||Q) = \frac{1}{\alpha-1} \log\left(\E_{x\sim Q}\left[ \left(\frac{P(x)}{Q(x)}\right)^\alpha \right]\right)$ is the $\alpha$-R\'enyi divergence between the probability distributions $P$ and $Q$.
\end{definition}
\begin{lemma}
[An equivalence between zero-concentrated DP and ``classical'' DP]
\label{lemma:zCDP_equivalence}
Let $M : \calX^n \to \calY$ be a randomized mechanism. Then, the following results hold:
\begin{itemize}
\item If $M$ is $\frac{\eps^2}{2}$-zCDP, then, for all $\delta > 0$, $M$ is
$
(\frac{\eps^2}{2} + \eps \sqrt{2\log(1/\delta)}, \delta)
$-DP.
\item If $M$ is $(\eps, 0)$-DP, then $M$ is $\frac{\eps^2}{2}$-zCDP.
\end{itemize}
\end{lemma}

\paragraph{Stein's Lemma and GLMs.}\label{sec:stein}
Finally, we state a multivariate version of Stein's Lemma \citep{stein1981estimation} due to \cite{liu1994siegel}.
\begin{lemma}
[Stein's Lemma \citep{liu1994siegel}]
\label{lemma:stein-lemma}
Let $\vec Z \in \reals^{p}, \vec W \in \reals^{q}$ be jointly Gaussian random vectors and let $f : \reals^{q} \rightarrow \reals$ be differentiable almost everywhere with 
$\E_{\vec W} \left[ | \partial f(\vec W) / \partial W_i | \right] < \infty$ for any $i \in [q]$. Then, $
\Cov [\vec Z, f(\vec W)] = \Cov[\vec Z, \vec W] \E[\nabla f(\vec W)]\,.
$
\end{lemma}
The lemma has a direct application on Generalized Linear Models with Gaussian covariates \citep{kadioglu2021sample,muratErdogdu, brillinger2012generalized, brillinger2012identification}:
for the GLM of \cite{brillinger2012generalized} with Gaussian covariates, whose model function satisfies the conditions of \Cref{lemma:stein-lemma}, one can show \citep{brillinger2012generalized, brillinger2012identification} that the Ordinary Least Squares Estimator 
asymptotically converges to the true parameter vector $\vec \beta$ of the GLM with probability $1$, up to a scaling factor $k$. This
is analogous to the scaling factor $k$ that appears in our analysis (see \Cref{theorem:main-result2}).

\section{Logistic Regression and SVMs as Models Satisfying \Cref{asm:glm}}
\label{appendix:discussion-glm-svm-log}

We will show here how the models of Logistic Regression and linearly-separable SVMs fit into our probabilistic model of Binary Regression (see \Cref{asm:glm}). For the Logistic Regression model, applying the model function $f(x) = 1/(1+e^{-x})$ directly yields the conditional probabilistic model of $\Pr[Y = 1 | \vec X] = \frac{1}{1+e^{-\vec\beta^T \vec X}}$ for the regression coefficients $\vec\beta$, which is precisely the desired Logistic Regression model.

In the second case, we consider linearly-separable SVMs, where the data $(\vec X_i, y_i)$ are completely separated by an underlying hyperplane that we are trying to uncover. That is, the model function would be $\sgn(\vec\beta^T \vec X)$. However, such a function is neither smooth nor continuously differentiable near the origin, therefore we will apply the following trick: after receiving the perfectly linearly separable data, we will induce a minuscule amount of noise through a noisy model function $f$ that is continuously differentiable and smooth everywhere (but crucially, near the origin). Intuitively, we smooth out the sign function. We can make infinitely good approximations of the sign function, and therefore passing the data through one of those before we apply our algorithm is sufficient to recover the true underlying $\vec\beta$ up to a scaling factor depending on the noise that we artificially introduced.

In order to overcome the model's non-smooth property, one way to smoothen this objective  is to approximate the sign function of the model by a sigmoid 
\[
f(x) = \frac{2}{1 + \exp(-\lambda x)} - 1\,,
\]
which depends on a parameter $\lambda$. This sigmoid function would then be smooth and continuously differentiable near the origin, as desired. Our method could then be applied (in a similar way to the logistic regression). We note that the parameter  would not affect the direction of  to be estimated, but would affect the accuracy guarantee obtained.

\section{Overview of \cite{gautamHighDimensional} and Extension to Sub-Gaussian Regime}
\label{appendix:variants}

We first review the modified covariance estimation result, and present a generalization of these results to the sub-gaussian case, whereupon we provide a subsequent discussion.

\subsection{Equivalence of Privacy Guarantees of \cite{gautamHighDimensional} to Classical DP}
\label{section:equiv-priv}
First of all, we note that the privacy guarantees hold for a variation of the differential privacy definition that is mentioned in \Cref{def:zCDP}; specifically, $\frac{\eps^2}{2}$ ``zero-concentrated DP.'' This variation is more lax than the classical (pure) $\eps$-DP, but stronger than the $(\eps, \delta)$-DP that is commonly used in the privacy literature, as can be seen immediately from \Cref{lemma:zCDP_equivalence}.

In our case, we prefer to keep the results of \Cref{theorem:main-result1} and \Cref{theorem:main-result2} in the classical $(\eps, \delta)$-DP definition, and therefore the respective algorithms are $(\frac{\eps^2}{2} + \eps\sqrt{2\log(1/\delta)}, \delta)$-DP. In essence, this guarantee is equivalent to $(\eps, \delta)$-DP, but this equivalence is worse as $\eps$ gets larger, i.e., in the case that little privacy is desired. Additionally, the fact that $\delta > 0$ allows us to bypass the requirement of an upper bound on $\|\vec\mu\|_2$ (role which was previously played by $R$), hence not requiring knowledge of $R$, as will be analyzed in \Cref{lemma:gautam_variant_mean}.

\subsection{Algorithm Overview}
\label{subsubsection:algo_variant_overview}

Here, we provide a high-level description of the algorithm $\textsc{LearnSubGaussian-hd}$ (used in \Cref{algo:betahat}) that learns the mean and covariance matrix of a high-dimensional Gaussian distribution, which differs slightly on its execution from the $\textsc{LearnGaussian-hd}$ algorithm in a way that will be analyzed hereafter.

The building block of the covariance estimation algorithm is the \textsc{NaivePCE} algorithm (Algorithm 1 of the work), which intuitively would be the first try at inducing privacy in the covariance estimation procedure. More specifically, it truncates the input samples, adds a random Gaussian matrix to the empirical covariance that arises from these samples, and outputs the projection of the final matrix to the PSD cone. However, this naive ``first try'' algorithm exhibits a linear dependence of the accuracy to the largest eigenvalue $\kappa$ of the covariance matrix $\Sigma$ (intuitively, the largest variance across any direction), whereas we aim for a $\log\kappa$ dependence. Therefore, noticing that the accuracy dependence is optimal when the aforementioned largest eigenvalue is of constant order, we seek to transform the samples $\vec X_i$ to $A \vec X_i$ such that the largest eigenvalue of the covariance matrix of $A \vec X_i$ (which is $A\Sigma A$ for symmetric matrices $A$) satisfies the above condition.

The covariance estimation algorithm, thus, begins by efficiently finding such a matrix $A$ (the ``preconditioner'') according to an algorithm (Algorithm 3 of the work) which does the following: it uses $O(\log\kappa)$ successive rounds of the \textsc{NaivePCE} algorithm such that every round ``eliminates'' the eigendirections of largest variance (through an eigenvector decomposition and keeping intact for the next rounds only the eigenvalues that are smaller than half the current upper bound) hence transforming each successive $\kappa_j$ (for $1\leq j\leq O(\log\kappa)$ the number of the current round) to $\kappa_{j+1} = 0.7 \kappa_j$. After $O(\log\kappa)$ rounds, the final largest eigenvalue of $A\Sigma A$ will be of constant order, as desired. After this procedure which finds $A$, \textsc{NaivePCE} is run on the samples $A \vec X_i$ with a result of $\wt{\Sigma}$, and the covariance estimation algorithm finally outputs $\wh{\Sigma} = A^{-1} \wt{\Sigma} A^{-1}$. For further consideration on the internal details of those algorithms, we refer the interested reader to \cite{gautamHighDimensional}. Note that for these steps, the knowledge of $\kappa$, the upper bound on the largest eigenvalue of the covariance matrix $\Sigma$ of the initial samples $\vec X_i$ is necessary for the calibrated truncation and noise addition to occur correctly.

Had someone wanted to also estimate the mean, they would first get a matrix $A$ as above through $2n$ samples $\frac{1}{\sqrt{2}} \left( \vec X_{2i} - \vec X_{2i-1} \right),\, 1\leq i\leq n$ that are i.i.d. with the same covariance matrix $\Sigma$, and then draw $n$ additional i.i.d. samples $\vec X_i$ (for a total of $3n$ samples) and apply the univariate mean estimation algorithm of \cite{vadhanConfidenceIntervals} to each coordinate of $A\vec X_i$ separately.
Our algorithm's difference with \cite{gautamHighDimensional} is that we call the algorithm of \cite{vadhanConfidenceIntervals} with $R=\infty$, which is allowable and efficient to do in our setting due to the privacy guarantee having $\delta > 0$ (see the guarantees on \Cref{appendix:variants-subgaussian}).
Once we are in a univariate sub-gaussian setting, and since we have a constant-order upper bound on the variance $\sigma^2 = O(1)$ of $\left( A\vec X_i \right)_j$, which denotes the $j$-th coordinate of the random vector $A\vec X_i$, the algorithm for univariate mean estimation works as follows:
First, we find a differentially private estimation of an upper bound $B$ on the data with high probability in the following way: we split the whole range that the mean might be located $(-\infty, \infty)$ to bins of width $\sigma = O(1)$, and taking advantage of the concentration of sub-gaussian random variables around their mean, we use a differentially private histogram algorithm \citep{stability_based_histograms_bun} to locate the most frequent bin, which (along with its neighboring bins) should contain all data points with high probability.
Second, we truncate the input data $\left( A\vec X_i \right)_j$ to a range calculated according to the above estimated bound, such that all input samples fall within that range with high probability, and then add Laplacian noise (calibrated according to the differentially-private calculated bound $B$) to the empirical mean of the input samples. We output as the result of the univariate mean estimation algorithm this noisy empirical mean of the (truncated) input samples.

Assuming the generic description of the algorithms above, we show how we extend the proofs to the sub-gaussian case below.

\subsection{Differentially Private Sub-Gaussian Mean and Covariance Estimation}
\label{appendix:variants-subgaussian}

The modified algorithms that we presented in \Cref{subsubsection:algo_variant_overview} have the following guarantees, whereupon we will provide a proof sketch.

\begin{lemma}
[Private Covariance Estimation]
\label{lemma:gautam_variant_covariance}
For every $\eps, \delta, \gamma, \kappa, \alpha > 0 \, ,$ there exists an $(\frac{\eps^2}{2} + \eps\sqrt{2\log(1/\delta)}, \delta)$-DP algorithm that, when given $n$ i.i.d. samples $\vec X_1, \dots, \vec X_n$ from a sub-gaussian multivariate distribution with mean $\E[\vec X_i] = \vec 0$ and covariance matrix $\E[\vec X_i \vec X_i^T] = \Sigma$ with $\mathbb{I}_d \preceq \Sigma \preceq \kappa \mathbb{I}_d$ and
\[
n =
O\left(
\frac{d + \log(1/\gamma)}{\alpha^2}
+ \frac{d^{3/2} \polylog\left( \frac{d}{\alpha\gamma\eps} \right)} {\alpha\eps}
+ \frac{d^{3/2} \sqrt{\log\kappa} \polylog\left( \frac{d\log\kappa} {\gamma\eps} \right) }{\eps}
\right)
\, ,
\]
outputs $\wh{\Sigma}$ such that $\left\| \Sigma^{-1/2} \left( \wh{\Sigma} - \Sigma \right) \Sigma^{-1/2} \right\|_2 \leq O(\alpha)$ with probability $1 - O(\gamma)$.
\end{lemma}

\begin{lemma}
[Private Mean Estimation]
\label{lemma:gautam_variant_mean}
For every parameters $\eps, \delta, \gamma, \kappa, \alpha > 0 \, ,$ there exists an $(\frac{\eps^2}{2} + \eps\sqrt{2\log(1/\delta)}, \delta)$-DP algorithm that, when given $n$ i.i.d. samples $\vec X_1, \dots, \vec X_n$ from a sub-gaussian multivariate distribution with mean $\E[\vec X_i] = \vec\mu$ and covariance matrix $\E[\vec X_i \vec X_i^T] = \Sigma$ with $\mathbb{I}_d \preceq \Sigma \preceq \kappa \mathbb{I}_d$ and
\[
n = O\left(
\frac{d \log( \frac{d}{\gamma} )}{\alpha^2}
+ \frac{d\polylog( \frac{d\log(1/\delta)} {\alpha\gamma\eps} )}{\alpha\eps}
+ \frac{\sqrt{d}\log( \frac{d}{\gamma\delta} )}{\eps}
+ \frac{d^{3/2} \sqrt{\log\kappa} \polylog\left( \frac{d\log\kappa} {\gamma\eps} \right)}{\eps}
\right)
\, ,
\]
outputs a (symmetric) matrix $A$ and a vector $\wh{\vec\mu}$ such that $\mathbb{I}_d \preceq A\Sigma A \preceq 1000\mathbb{I}_d$ and $\left\| A ( \wh{\vec\mu} - \vec\mu ) \right\|_2 \leq \alpha$ with probability $1 - O(\gamma)$.
\end{lemma}

First of all, the respective algorithms adumbrated in \Cref{subsubsection:algo_variant_overview} hold for the case of sub-gaussian input random vectors too, because the concentration bounds that are utilized readily generalize to the sub-gaussian case. We provide here the variants of the concentration bounds that are needed for these algorithms, and
we then show how the second modification with respect to the consideration of $R=\infty$ (see \Cref{subsubsection:algo_variant_overview} for this modification) alters the guarantees provided.

By \cite{dkk}, we have the following generalizations of concentration bounds in the sub-gaussian regime:

\begin{lemma}
\label{lemma:max_norm_subgaussian}
Let $\vec X_1, \dots, \vec X_n \in \reals^d$ be $n$ i.i.d. samples from a sub-gaussian multivariate distribution with mean $\E[\vec X_i] = \vec 0$ and covariance matrix $\E[\vec X_i \vec X_i^T] = \Sigma$. Then, with probability $1-O(\gamma)$, it holds that
\[
\left\| \Sigma^{-1/2} \vec X_i \right\|_2^2
\leq
d\log(n/\gamma)
,\, \forall i\in [n]
\, .
\]
\end{lemma}

\begin{lemma}
[Sub-gaussian covariance matrix estimation]
Let $\vec X_1, \dots, \vec X_n \in \reals^d$ be $n$ i.i.d. samples from a sub-gaussian multivariate distribution with mean $\E[\vec X_i] = \vec 0$ and covariance matrix $\E[\vec X_i \vec X_i^T] = \Sigma$. Define $\vec Z_i = \Sigma^{-1/2} \vec X_i$ with covariance matrix $\E[\vec Z_i \vec Z_i^T] = \mathbb{I}_d$. Then, with probability $1-O(\gamma)$, all the following hold:
\begin{align*}
\left\| \frac{1}{n} \sum_{i=1}^n \vec Z_i \vec Z_i^T - \mathbb{I}_d \right\|_2 &\leq O\left( \sqrt{\frac{d+\log(1/\gamma)}{n}} \right)
\\
\left( 1 - O\left( \sqrt{\frac{d+\log(1/\gamma)}{n}} \right) \right) \cdot \mathbb{I}_d \preceq \frac{1}{n} &\sum_{i=1}^n \vec Z_i \vec Z_i^T \preceq \left( 1 + O\left( \sqrt{\frac{d+\log(1/\gamma)}{n}} \right) \right) \cdot \mathbb{I}_d
\\
\left\| \frac{1}{n} \sum_{i=1}^n \vec Z_i \vec Z_i^T - \mathbb{I}_d \right\|_F &\leq O\left( \sqrt{\frac{d^2+\log(1/\gamma)}{n}} \right)
\end{align*}
where $\left\| A \right\|_F$ is the Frobenius norm of a matrix $A$, defined as the square root of the sum of the squares of each of its entries.
\end{lemma}

Note that the necessary bounds for the univariate case are obtained simply by setting $d=1$ in the above lemmata. These are required for adapting the proofs of \cite{vadhanConfidenceIntervals} to the sub-gaussian regime.

Additionally, we remark that we do not need to modify the concentration bound arising from the Hanson-Wright inequality for bounding the norm of the noise matrix that is added according to \textsc{NaivePCE}, since, even in the case that the inputs are sub-gaussian, the added noise is purely Gaussian.

\begin{proof}[Proof Sketch]

We now provide a proof sketch for the modification of the proof of \cite{gautamHighDimensional}. In this sketch, we will use the standard notation to move forward with our variation of the lemmata.

First, we provide an alternative lemma that removes the dependency on a prior bound for $\|\vec\mu\|_2$ when the DP guarantee that we desire to achieve is $\delta > 0$.

\begin{lemma}
\label{lemma:vadhan_variation}
For every $\eps, \delta, \gamma, \kappa, \alpha > 0$, there exists an $(\eps, \delta)$-DP algorithm that, when given $n$ i.i.d. samples $X_1, \dots, X_n$ from a sub-gaussian univariate distribution with mean $\E[X_i] = \mu$ and variance $\E[(X_i - \mu)^2] = \sigma^2$ with $1 \leq \sigma^2 \leq \kappa$ and
\[
n = O\left(
\frac{\log(1/\gamma)}{\alpha^2}
+ \frac{\polylog\left( \frac{\log(1/\delta)}{\alpha\gamma\eps} \right)}{\alpha\eps}
+ \frac{\log(1/\delta) + \log(1/\gamma)}{\eps}
\right)
\, ,
\]
outputs $\wh{\mu}$ such that $|\wh{\mu} - \mu| \leq \alpha\kappa$ with probability $1-\gamma$.
\end{lemma}

Generalizing this algorithm to the multivariate case, by following the $\textsc{NaivePME}$ algorithm referenced in \Cref{subsubsection:algo_variant_overview}, we obtain the following lemma, in the proof sketch of which we show only the modifications required.

\begin{lemma}
\label{lemma:vadhan_generalization_with_bad_kappa}
For every $\eps, \delta, \gamma, \kappa, \alpha > 0$, there exists an $(\frac{\eps^2}{2} + \eps\sqrt{2\log(1/\delta)}, \delta)$-DP algorithm that, when given $n$ i.i.d. samples $\vec X_1, \dots, \vec X_n$ from a sub-gaussian multivariate distribution with mean $\E[\vec X_i] = \vec\mu$ and covariance matrix $\E[\vec X_i \vec X_i^T] = \Sigma$ with $\mathbb{I}_d \preceq \Sigma \preceq \kappa \mathbb{I}_d$ and
\[
n = O\left(
\frac{\kappa^2 d\log(d/\gamma)}{\alpha^2}
+ \frac{\kappa d \polylog\left( \frac{\kappa d\log(1/\delta)} {\alpha\gamma\eps} \right)}{\alpha\eps}
+ \frac{\sqrt{d} \left( \log(1/\delta) + \log(d/\gamma) \right)}{\eps}
\right)
\, ,
\]
outputs $\wh{\vec\mu}$ such that $\| \wh{\vec\mu} - \vec\mu \|_2 \leq \alpha$ with probability $1-\gamma$.
\end{lemma}
\begin{proof}[Proof Sketch]
Following the procedure of \Cref{lemma:vadhan_variation} for each dimension of the multivariate vectors $\vec X_i$ with the appropriate parameter settings as described in $\textsc{NaivePME}$, we obtain the following final result:
\[
\| \wh{\vec\mu} - \vec\mu \|_2 \leq \sqrt{d} \max_{1\leq i\leq d} |\wh{\mu}_i - \mu_i| \leq \sqrt{d} \left(\frac{\alpha}{\sqrt{d}}\right) = \alpha
\, ,
\]
as desired.
\end{proof}

The final step is to use the private ``preconditioner'' $A$ in order to reduce the condition number of $\Sigma$ (which is at most $\kappa$) to at most a constant, and obtain the final bound of the lemma that we seek. Again, we show the modification of the bound, hinging on \Cref{lemma:vadhan_generalization_with_bad_kappa} which is used to output the final estimate $\wh{\vec\mu} = A^{-1} \wt{\vec\mu}$ from the estimate $\wt{\vec\mu}$ of the mean of the variables $A\vec X_i$, as follows by the lemma with $\kappa=O(1)$:

\[
\left\| A ( \wh{\vec\mu} - \vec\mu ) \right\|_2 = \left\| \wt{\vec\mu} - A\vec\mu \right\|_2 \leq \alpha
\, .
\]
\end{proof}

\section{Proof of \Cref{theorem:main-result1}}
\label{appendix:proof-main-result1}

We divide the proof of \Cref{theorem:main-result1} in a series of claims. For convenience, we restate the (stronger) version of the Theorem that we will prove here:

\begin{theorem}
[Privacy and Accuracy of $\wh{\vec\beta}$ in Private Least Squares Fitting]
Under \Cref{asm:gauss1} with parameters $(\kappa, c, \rho, R)$, for all privacy parameters $\eps, \delta > 0$, accuracy parameters $\alpha, \eta > 0$ and confidence $\gamma \in (0,1)$, \textsc{PrivLearnLSE} (defined in \Cref{algo:betahat}) is $(\frac{\eps^2}{2} + \eps\sqrt{2\log(1/\delta)}, \delta)$-differentially private.
Moreover, if the number of labeled examples is at least: 
\begin{align}
n &= O\left(
\frac{d \log( \frac{d}{\gamma} )}{\eta^2}
+ \frac{d\polylog( \frac{d\log(1/\delta)} {\eta\gamma\eps} )}{\eta\eps}
+ \frac{d^{3/2} \sqrt{\log(\kappa\rho c)} \polylog\left( \frac{d\log(\kappa\rho c)} {\gamma\eps\delta} \right)}{\eps}
\right)
\nonumber \\
& + O\left(
(1+R) \left(
\frac{d \log( \frac{d}{\gamma} )}{\alpha^2}
+ \frac{d^{3/2} \polylog\left( \frac{d\log(1/\delta)}{\alpha\gamma\eps} \right)} {\alpha\eps}
+ \frac{d^{3/2} \sqrt{\log\kappa} \polylog\left( \frac{d\log\kappa} {\gamma\eps} \right) }{\eps}
\right)
\right)
\, ,
\label{eq:ubound}
\end{align}
then with probability at least $1-O(\gamma)$ an estimate $\wh{\vec\beta} \in \reals^d$ is successfully output and along with the LSE $\vec\beta^*$ satisfies:
\begin{align}
\label{eq:final_ineq}
\left\| \wh{\vec\beta} - \vec\beta^\star \right\|_2^2
\leq \left\| \wh{\vec w} - \vec w^\star \right\|_2^2
\leq O\left( \alpha^2 \right)
\cdot \left\|\vec w^\star\right\|_2^2
+ O\left( \eta^2 \right) \cdot c^2
\, ,
\end{align}
where $\wh{\vec w} = \Sigma^{1/2} \wh{\vec\beta} ~\text{and}~ \vec w^\star = \Sigma^{1/2} \vec\beta^{\star}$. Finally, \textsc{PrivLearnLSE} runs in $\poly(n)$ time.
\end{theorem}
The outline of the proof is as follows. First, we prove the privacy guarantee (see \Cref{subsection:priv_guarant}); the latter follows in a similar fashion as the privacy proofs of the algorithms in \Cref{appendix:variants-subgaussian}, exploiting the extension to $\delta>0$ that we described, and generalizing the obtained privacy for altering both $y_i$ and $\vec X_i$ (see \Cref{section:technical-overview}). In the case of accuracy (see \Cref{app:acc}), we first establish the following inequality:
\[
\left\| \wh{\vec w} - \vec w^\star \right\|_2^2 \leq 2 \left\| \Sigma^{1/2} \left(\wh{\Sigma} + \wh{\vec\mu} \wh{\vec\mu}^T \right)^{-1} \Sigma^{1/2} \right\|_2^2 \left( \left\| \Sigma^{-1/2} Q_1 \Sigma^{-1/2} \right\|_2^2 \cdot \left\| \vec w^\star \right\|_2^2 + \left\| \Sigma^{-1/2} \vec Q_2 \right\|_2^2 \right)
\, .
\]
Via a series of claims (see \Cref{claim:inverse_sigma_norm}, \Cref{claim:Q1_norm}, and \Cref{claim:Q2_norm}) we bound each of the constituent terms of the right-hand-side of this inequality, finally yielding \Cref{eq:final_ineq} with as many samples as in \Cref{eq:ubound}.

\subsection{Proof of Privacy Guarantee}
\label{subsection:priv_guarant}

We show first that
\Cref{algo:betahat}, that computes $\wh{\vec\beta} \in \reals^d$ as in \Cref{theorem:main-result1}, is $(\frac{\eps^2}{2} + \eps\sqrt{2\log(1/\delta)}, \delta)$-DP.
Our dataset consists of $n$ pairs $(\vec X_i, y_i) \in \reals^d \times \reals$, therefore we will look into what happens if one of those pairs is altered: specifically, consider that the pair $(\vec X_i, y_i)$ becomes $(\vec X_i', y_i')$ for some (specific) $i$. The algorithm for $\wh{\vec\beta}$ uses three sub-algorithms, which we claim will be $\left( \frac{1}{3} \left(\frac{\eps^2}{2} + \eps\sqrt{2\log(1/\delta)} \right), \frac{\delta}{3} \right)$-DP each (as described in the algorithm above, it suffices to consider $O(\eps)$ and $O(\delta)$ as parameters of each one), thereby giving us the final result of the claim by the advanced composition properties of differentially private mechanisms (\Cref{def:dp-composition}).

For the covariance estimation algorithm (used both in the covariance estimation of the feature vectors $\vec X_i$ and in the mean estimation of the product $y_i \vec X_i$), it now suffices to show that the interface of the algorithms used in \Cref{subsubsection:algo_variant_overview} with the underlying data, i.e., the \textsc{NaivePCE} algorithm, is differentially private. Indeed, this result arises by computing the sensitivity of the (truncated) empirical covariance in the following way:
\[
\left\| \frac{1}{n} \left( y_i^2 \vec X_i \vec X_i^T - y_i'^2 \vec X_i' \vec X_i'^T \right) \right\|_F
\leq
\frac{1}{n} \left( \left\| y_i \vec X_i \right\|_2^2 + \left\| y_i' \vec X_i' \right\|_2^2 \right)
\leq
O\left( \frac{d \kappa c^2 \log(n/\gamma)}{n} \right)
\, ,
\]
since the truncation happens in accordance with \Cref{lemma:max_norm_subgaussian}, thereby allowing us to add the Gaussian noise of the magnitude prescribed in \textsc{NaivePCE}. Hence, the rest of the algorithms which hinge on \textsc{NaivePCE} are differentially private, due to the differential privacy composition theorems (see \Cref{def:dp-composition}).

In a similar fashion, the algorithm for mean estimation in \cite{vadhanConfidenceIntervals} depends upon the stability-based histogram learner of \cite{stability_based_histograms_bun} which is in turn based on the idea of introducing Laplacian noise to the empirical histogram of a dataset. Denoting $X_{ij}$ as the $j$-th coordinate of the feature vector $\vec X_i$, we notice that the sensitivity of the empirical counting function is $2/n$ regardless of the input variations (at most 2 bins could have their counts altered in the worst case, if some $y_i X_{ij}$ changed its location from the bin it was to another, whereas all the other products had the same locations in bins). This sensitivity is, thus, independent of whether both $y_i$ and $X_{ij}$ were changed in our model, thereby affording the desired level of privacy to the whole algorithm. \qed

\subsection{Proof of Accuracy Guarantee}\label{app:acc}

For simplicity in notation, in what follows we notate $\vec\mu = \vec\mu_{\vec X}$ and $\vec\mu' = \vec\mu_{\vec X, y}$ (likewise, $\wh{\vec\mu} = \wh{\vec\mu}_{\vec X}$ and $\wh{\vec\mu}' = \wh{\vec\mu}_{\vec X, y}$).

We begin by adding and subtracting the quantities of each factor of $\vec\beta^{\star} \in \reals^d$, as follows:
\[
\wh{\vec\beta} - \vec\beta^{\star} = \left(\wh{\Sigma} + \wh{\vec\mu} \wh{\vec\mu}^T \right)^{-1} \left( -Q_1 \vec\beta^{\star} + \vec Q_2 \right)
\Leftrightarrow
\left(\wh{\Sigma} + \wh{\vec\mu} \wh{\vec\mu}^T \right) \left( \wh{\vec\beta} - \vec\beta^{\star} \right) = -Q_1 \vec\beta^{\star} + \vec Q_2
\, ,
\]
where
\[
Q_1 = \wh{\Sigma} + \wh{\vec\mu} \wh{\vec\mu}^T - \frac{1}{n} X^T X
,~~~\text{and}~~~
\vec Q_2 = \wh{\vec\mu}' - \frac{1}{n} X^T \vec y
\, .
\]

Then, substituting the quantities $\vec w = \Sigma^{1/2} \vec\beta$ and moving terms from the left side to the right of the equation, it holds that
\begin{align*}
& \left(\wh{\Sigma} + \wh{\vec\mu} \wh{\vec\mu}^T \right) \left( \wh{\vec\beta} - \vec\beta^{\star} \right) = -Q_1 \vec\beta^{\star} + \vec Q_2 \\
\Leftrightarrow\ & \left(\wh{\Sigma} + \wh{\vec\mu} \wh{\vec\mu}^T \right) \Sigma^{-1/2} \left( \wh{\vec w} - \vec w^\star \right) = -Q_1 \Sigma^{-1/2} \vec w^\star + \vec Q_2 \\
\Leftrightarrow\ &\Sigma^{-1/2} \left(\wh{\Sigma} + \wh{\vec\mu} \wh{\vec\mu}^T \right) \Sigma^{-1/2} \left( \wh{\vec w} - \vec w^\star \right) = -\Sigma^{-1/2} Q_1 \Sigma^{-1/2} \vec w^\star + \Sigma^{-1/2} \vec Q_2 \\
\Leftrightarrow\ &\wh{\vec w} - \vec w^\star = \Sigma^{1/2} \left(\wh{\Sigma} + \wh{\vec\mu} \wh{\vec\mu}^T \right)^{-1} \Sigma^{1/2} \left( -\Sigma^{-1/2} Q_1 \Sigma^{-1/2} \vec w^\star + \Sigma^{-1/2} \vec Q_2 \right)
\, .
\end{align*}

Using Cauchy-Schwartz and the sub-multiplicative property of the spectral norm, we establish the following inequality:
\[
\left\| \wh{\vec w} - \vec w^\star \right\|_2^2 \leq 2 \left\| \Sigma^{1/2} \left(\wh{\Sigma} + \wh{\vec\mu} \wh{\vec\mu}^T \right)^{-1} \Sigma^{1/2} \right\|_2^2 \left( \left\| \Sigma^{-1/2} Q_1 \Sigma^{-1/2} \right\|_2^2 \cdot \left\| \vec w^\star \right\|_2^2 + \left\| \Sigma^{-1/2} \vec Q_2 \right\|_2^2 \right)
\, .
\]
We state three claims that bound the constituent terms of the right-hand-side of this inequality:
\begin{claim}
\label{claim:inverse_sigma_norm}
When
\[
n = \Omega\left(
d
+ \log(1/\gamma)
+ \frac{d^{3/2} \sqrt{\log\kappa} \polylog\left( \frac{d\log\kappa} {\gamma\eps} \right) }{\eps}
\right)
\, ,
\]
the following inequality holds with probability $1-O(\gamma)$: \[
\left\| \Sigma^{1/2} \left(\wh{\Sigma} + \wh{\vec\mu} \wh{\vec\mu}^T \right)^{-1} \Sigma^{1/2} \right\|_2^2
\leq
O\left( 1 \right)
\, .
\]
\end{claim}

\begin{claim}
\label{claim:Q1_norm}
For every $\alpha > 0$, when
\[
n = \Omega\left(
(1+R) \left(
\frac{d \log( \frac{d}{\gamma} )}{\alpha^2}
+ \frac{d^{3/2} \polylog\left( \frac{d\log(1/\delta)}{\alpha\gamma\eps} \right)} {\alpha\eps}
+ \frac{d^{3/2} \sqrt{\log\kappa} \polylog\left( \frac{d\log\kappa} {\gamma\eps} \right) }{\eps}
\right)
\right)
\, ,
\]
the following inequality holds with probability $1-O(\gamma)$:
\[
\left\| \Sigma^{-1/2} Q_1 \Sigma^{-1/2} \right\|_2^2 \leq O(\alpha^2)
\, .
\]
\end{claim}

\begin{claim}
\label{claim:Q2_norm}
For every $\eta > 0$, when
\[
n = \Omega\left(
\frac{d \log( \frac{d}{\gamma} )}{\eta^2}
+ \frac{d\polylog( \frac{d\log(1/\delta)} {\eta\gamma\eps} )}{\eta\eps}
+ \frac{\sqrt{d}\log( \frac{d}{\gamma\delta} )}{\eps}
+ \frac{d^{3/2} \sqrt{\log(\kappa\rho c)} \polylog\left( \frac{d\log(\kappa\rho c)} {\gamma\eps} \right)}{\eps}
\right)
\, ,
\]
the following inequality holds with probability $1-O(\gamma)$:
\[
\left\| \Sigma^{-1/2} \vec Q_2 \right\|_2^2 \leq O(\eta^2) \cdot c^2
\, .
\]
\end{claim}
 We prove each of these claims individually below (see \Cref{sec:proofofc1}--\Cref{sec:proofofc3}).
When combined with a union bound of the respective probabilistic events, these claims give the desired \Cref{theorem:main-result1}. In particular,
we directly obtain \Cref{theorem:main-result1}, since 
\[
\left\| \wh{\vec\beta} - \vec\beta^\star \right\|_2^2
= \left\| \Sigma^{-1/2} \Sigma^{1/2} \left( \wh{\vec\beta} - \vec\beta^\star \right) \right\|_2^2
\leq \left\| \Sigma^{-1/2} \right\|_2^2
\cdot \left\| \wh{\vec w} - \vec w^\star \right\|_2^2
\leq \left\| \wh{\vec w} - \vec w^\star \right\|_2^2
\, ,
\]
by the sub-multiplicative property of the norm and because $\mathbb{I}_d \preceq \Sigma$. \qed

\subsubsection{Proof of \Cref{claim:inverse_sigma_norm}} \label{sec:proofofc1}

Following the covariance estimation procedure, we recall the ``private preconditioner'' matrix $A$ that is used to reduce the effect of the condition number of $\Sigma$ from at most $O(\kappa)$ to at most a constant order factor. More specifically, the following lemma holds:

\begin{lemma}
[Theorem 3.11 of \cite{gautamHighDimensional}]
\label{lemma:preconditioner}
For every $\eps, \delta, \gamma, \alpha, \kappa > 0 \, ,$ there exists an algorithm that, when given $n$ i.i.d. samples $\vec X_1, \dots, \vec X_n$ from a sub-gaussian multivariate distribution with mean $\E[\vec X_i] = \vec\mu$ and covariance matrix $\E[\vec X_i \vec X_i^T] = \Sigma$ with $\mathbb{I}_d \preceq \Sigma \preceq \kappa \mathbb{I}_d$ and
\[
n = O\left(
\frac{d^{3/2} \sqrt{\log\kappa} \polylog\left( \frac{d\log\kappa} {\gamma\eps} \right) }{\eps}
\right)
\, ,
\]
outputs a (symmetric) matrix $A$ (the ``private preconditioner'') such that $ \mathbb{I}_d \preceq A \Sigma A \preceq 1000 \mathbb{I}_d$ with probability $1 - O(\gamma)$.
\end{lemma}

Afterwards, the naive private estimation by addition of Gaussian noise through a random Gaussian matrix perturbation to the sample covariance matrix estimate is run, taking as input the ``normalized'' samples $A\vec X_i$, and by denoting $\wt{\Sigma}$ this estimate (which is explained in the proof of \Cref{fact:sigma_tilde_inverse_norm}) and also $\wt{\vec\nu} = A \wh{\vec\mu}$, we have that $\wh{\Sigma} = A^{-1} \wt{\Sigma} A^{-1}$ and therefore with probability $1-O(\gamma) \, ,$
\[
\left\| \Sigma^{1/2} \left(\wh{\Sigma} + \wh{\vec\mu} \wh{\vec\mu}^T \right)^{-1} \Sigma^{1/2} \right\|_2^2
\leq
\left\| \Sigma^{1/2} A \right\|_2^2 \cdot \left\| \left( \wt{\Sigma} + \wt{\vec\nu} \wt{\vec\nu}^T \right)^{-1} \right\|_2^2 \cdot \left\| A \Sigma^{1/2} \right\|_2^2
\, ,
\]
which gives the desired result, due to the following two facts.

\begin{fact}
\label{fact:ASigma_norm}
With probability $1-O(\gamma)$, $ \left\| \Sigma^{1/2} A \right\|_2^2 = \left\| A \Sigma^{1/2} \right\|_2^2 \leq O(1) \, . $
\end{fact}
\begin{proof}
Since $A$ is a symmetric square matrix, $\left\| \Sigma^{1/2} A \right\|_2^2 = \left\| A \Sigma^{1/2} \right\|_2^2$. By using the definition of the spectral norm, the fact's statement is immediately obtained:
\[
\left\| \Sigma^{1/2} A \right\|_2^2
= \sigma_\text{max}^2 \left( \Sigma^{1/2} A \right)
= \lambda_\text{max} \left( \left(\Sigma^{1/2} A\right)^T \Sigma^{1/2} A \right)
= \lambda_\text{max} \left( A\Sigma A \right)
\leq 1000
\, ,
\]
since by \Cref{lemma:preconditioner}, $A \Sigma A \preceq 1000 \mathbb{I}_d$ with probability $1-O(\gamma)$.
\end{proof}

\begin{fact}
\label{fact:sigma_tilde_inverse_norm}
With probability $1-O(\gamma)$, when
\[
n = \Omega\left(
d+\log(1/\gamma)
+ \frac{\sqrt{d}\polylog\left( \frac{d}{\eps\gamma} \right)}{\eps}
\right)
\, ,
\]
it holds that
\[
\left\| \left( \wt{\Sigma} + \wt{\vec\nu} \wt{\vec\nu}^T \right)^{-1} \right\|_2^2
\leq 1 + O\left( \sqrt{\frac{d+\log(1/\gamma)}{n}} + \frac{\sqrt{d}\log(1/\gamma)\log(n/\gamma)}{n\eps} \right)
\, .
\]
\end{fact}
\begin{proof}

In order to prove this fact, we need to delve further into the procedure by which $\wt{\Sigma}$ is generated (see \Cref{subsubsection:algo_variant_overview}), i.e., $\textsc{NaivePCE}$. In short, we will use that $\wt{\Sigma}$ is the projection into the PSD cone of the empirical covariance matrix of the inputs (which are the vectors $A\vec X_1, \dots, A\vec X_n$ where $A$ is the above ``preconditioner'' matrix) plus a symmetric random matrix $N$ of small Gaussian perturbations (that serves to enforce the privacy guarantee). Hence, the following holds:
\[
\wt{\Sigma} = \text{proj}_\text{PSD} \left( \frac{1}{n}\sum_{i=1}^n \vec Z_i \vec Z_i^T + N \right)
\, ,
\]
where $\vec Z_i = A \vec X_i$, and the matrix $N$ is a symmetric random matrix with dimension $d\times d$ whose entries $N_{ij}, j\geq i$ are i.i.d. Gaussian random variables with zero mean and standard deviation $\sigma = \frac{d\log(n/\gamma)}{n\eps}$.

By Weyl's inequality for matrices, it is true for two real, symmetric matrices $A,B$ and their sum $A+B$ that $\lambda_\text{min}(A+B) \geq \lambda_\text{min}(A) + \lambda_\text{min}(B)$. Because $\left\| \left( \wt{\Sigma} + \wt{\vec\nu} \wt{\vec\nu}^T \right)^{-1} \right\|_2^2 = \frac{1}{\lambda_\text{min}^2 \left( \wt{\Sigma} + \wt{\vec\nu} \wt{\vec\nu}^T \right)}$, we will prove a lower bound about $\lambda_\text{min}\left( \wt{\Sigma} + \wt{\vec\nu} \wt{\vec\nu}^T \right) \geq \lambda_\text{min}\left( \wt{\Sigma} \right)$ (since $\wt{\vec\nu} \wt{\vec\nu}^T$ is a PSD matrix) in the following way: we will show a bound about the minimum eigenvalue of the inner sum $\frac{1}{n}\sum_{i=1}^n \vec Z_i \vec Z_i^T + N$, and argue that it is positive with high probability $1-O(\gamma)$. This means that the projection of this matrix into the PSD cone is the same as the matrix itself with high probability, therefore the bound on the minimum eigenvalue will hold verbatim.

We begin by tailoring a lemma from random matrix theory referenced in \cite{tao_random_matrices} to the random matrix $N$:
\begin{lemma}
[Concentration of symmetric random matrices with Gaussian entries]
Suppose that the entries $N_{ij}$ for $j\geq i$ of a symmetric matrix $N$ with dimensions $d\times d$ are i.i.d. Gaussian random variables with zero mean and variance $\sigma^2$. Then, there exist universal constants $C, c > 0$ such that the largest singular value of $N$ satisfies for all $A \geq C$:
\[
\Pr\left[ s_\text{max}(N) > A \sigma\sqrt{d} \right] \leq C \exp(-cAd)
\, .
\]
\end{lemma}

The above lemma means that with probability at least $1-\gamma$, we have that the largest singular value of $N$ is at most
\[
s_\text{max}(N) \leq O\left( \frac{\sigma\log(1/\gamma)}{\sqrt{d}} \right)
\, .
\]

Due to the matrix $N$ being square symmetric with dimensions $d\times d$, it holds that its singular values are the absolute values of its eigenvalues, therefore $|\lambda_\text{min}(N)|$ is one of the singular values of $N$ (note that it could even be the largest), hence $|\lambda_\text{min}(N)| \leq s_\text{max}(N)$ and thus
\begin{equation}
\label{eq:lambda_min_N}
\lambda_\text{min}(N)
\geq -s_\text{max}(N)
\geq -O\left( \frac{\sigma\log(1/\gamma)}{\sqrt{d}} \right)
= -O\left( \frac{\sqrt{d}\log(1/\gamma)\log(n/\gamma)}{n\eps} \right)
\, ,
\end{equation}
since we remind the reader that $\sigma = \frac{d\log(n/\gamma)}{n\eps}$.

Lower bounds on the minimum eigenvalue of sample covariance matrices of the form  are well-known in the literature, and we use here a version that appears in \cite{dkk}. Note that the vectors $\vec Z_i = A\vec X_i$, whose covariance matrix we are interested in, have a ``normalized'' distribution with covariance matrix $A\Sigma A^T = A\Sigma A$ (by the symmetry of $A$) that has at most a constant eigenvalue, since $A\Sigma A \preceq 1000\mathbb{I}_d$ by construction of the ``preconditioner'' $A$ with probability $1-O(\gamma)$. Therefore, by classical covariance matrix estimation inequalities for eigenvalues of sample covariance matrices from \cite{dkk}, it holds that with probability $1-O(\gamma)$,
\begin{equation}
\label{eq:lambda_min_sample_cov_matrix}
\lambda_\text{min} \left( \frac{1}{n}\sum_{i=1}^n \vec Z_i \vec Z_i^T \right)
\geq 1 - O\left( \sqrt{\frac{d+\log(1/\gamma)}{n}} \right)
\, .
\end{equation}

To conclude, we combine the lower bounds in eigenvalues of \Cref{eq:lambda_min_N} and \Cref{eq:lambda_min_sample_cov_matrix}. By Weyl's inequality, with probability $1-O(\gamma)$, we have that:
\begin{align*}
\lambda_\text{min} \left( \frac{1}{n}\sum_{i=1}^n \vec Z_i \vec Z_i^T + N \right)
&\geq
\lambda_\text{min}\left( \frac{1}{n}\sum_{i=1}^n \vec Z_i \vec Z_i^T \right)
+ \lambda_\text{min}\left( N \right) \\
&\geq
1 - O\left( \sqrt{\frac{d+\log(1/\gamma)}{n}} + \frac{\sqrt{d}\log(1/\gamma)\log(n/\gamma)}{n\eps} \right)
\, .
\end{align*}

Choosing $n$ such that the above lower bound is positive\footnote{We remark that this eigenvalue lower bound also means that the eigenvalue of $\widetilde{\Sigma} + \widetilde{\nu}\widetilde{\nu}^T = A \left( \widehat{\Sigma} + \widehat{\mu}\widehat{\mu}^T \right) A$ is bounded away from $0$ by means of the chosen sample size $n$. This directly implies, since $A$ is always invertible by \cite{gautamHighDimensional}, that $\mathrm{M} = \widehat{\Sigma} + \widehat{\mu}\widehat{\mu}^T$ in \Cref{algo:betahat} is also invertible with the same high probability, $1-O(\gamma)$.}, i.e., if
\[
n = \Omega\left(
d+\log(1/\gamma)
+ \frac{\sqrt{d}\polylog\left( \frac{d}{\eps\gamma} \right)}{\eps}
\right)
\, ,
\]
then the projection of the matrix sum into the PSD cone is equal to the matrix sum itself, therefore directly arriving at the final result by noting the additional fact that:
\[
\left( \frac{1}{1-O(x)} \right)^2 \leq 1 + O(x)
\, ,
\]
obtainable by a Taylor expansion since we chose $n$ to be at least such that the denominator of the fraction is positive.
\end{proof}

By combining \Cref{fact:ASigma_norm} and \Cref{fact:sigma_tilde_inverse_norm}, we arrive at the final result of \Cref{claim:inverse_sigma_norm}.

\subsubsection{Proof of \Cref{claim:Q1_norm}}
\label{sec:proofofc2}

Initially, breaking $Q_1$ as
\begin{align*}
Q_1
&= \wh{\Sigma} + \wh{\vec\mu} \wh{\vec\mu}^T - \frac{1}{n} X^T X \\
&= \left( \wh{\Sigma} - \Sigma \right) - \left( \frac{1}{n} \sum_{i=1}^n \left( \vec X_i - \vec\mu \right) \left( \vec X_i - \vec\mu \right)^T - \Sigma \right) \\
&~~~ + \left( \wh{\vec\mu} \wh{\vec\mu}^T - \vec\mu \vec\mu^T - \vec\mu \left( \frac{1}{n} \sum_{i=1}^n \vec X_i - \vec\mu \right)^T - \left( \frac{1}{n} \sum_{i=1}^n \vec X_i - \vec\mu \right) \vec\mu^T \right) \\
&= \left( \wh{\Sigma} - \Sigma \right) - \left( \frac{1}{n} \sum_{i=1}^n \left( \vec X_i - \vec\mu \right) \left( \vec X_i - \vec\mu \right)^T - \Sigma \right) - \vec\mu \left( \frac{1}{n} \sum_{i=1}^n \vec X_i - \vec\mu \right)^T - \left( \frac{1}{n} \sum_{i=1}^n \vec X_i - \vec\mu \right) \vec\mu^T \\
&~~~ + \left( \left( \wh{\vec\mu} - \vec\mu \right) \left( \wh{\vec\mu} - \vec\mu \right)^T + \left( \wh{\vec\mu} - \vec\mu \right) \vec\mu^T + \vec\mu \left( \wh{\vec\mu} - \vec\mu \right)^T \right)
\, ,
\end{align*}
it follows that
\begin{align*}
\left\| \Sigma^{-1/2} Q_1 \Sigma^{-1/2} \right\|_2^2
&\leq 2\left\| \Sigma^{-1/2} \left( \wh{\Sigma} - \Sigma \right) \Sigma^{-1/2} \right\|_2^2 \\
&~~~+ 2\left\| \Sigma^{-1/2} \left( \frac{1}{n} \sum_{i=1}^n {\left(\vec X_i - \vec\mu \right) \left( \vec X_i - \vec\mu \right)^T} - \Sigma \right) \Sigma^{-1/2} \right\|_2^2 \\
&~~~+ 2\left\| \vec\mu \right\|_2^2 \cdot \left( 2 \left\| \frac{1}{n} \sum_{i=1}^n \vec V_i \right\|_2^2 + O\left( \left\| \Sigma^{-1/2} \left( \wh{\vec\mu} - \vec\mu \right) \right\|_2^2 \right) \right)
\, ,
\end{align*}
where we have used the variance-normalized vectors $\vec V_i = \Sigma^{-1/2} \left( \vec X_i - \vec\mu \right)$ which have covariance matrix $\mathbb{I}_d$.

We state and prove the two below facts which, along with \Cref{lemma:gautam_variant_covariance}, which we remind to the reader that it is applied to the $2n$ sample differences $\frac{1}{\sqrt{2}}\left( \vec X_{2i} - \vec X_{2i-1} \right)$, so that they have zero mean, and with a similar procedure to \Cref{fact:ASigma_norm} and \Cref{lemma:gautam_variant_mean}, leads to the desired result immediately.

\begin{fact}
\label{fact:empirical_covariance}
For every $\alpha > 0$, with probability $1-O(\gamma)$, when
\[
n = \Omega\left( \frac{d+\log(1/\gamma)}{\alpha^2} \right)
\, ,
\]
it holds that
\[
\left\| \Sigma^{-1/2} \left( \frac{1}{n} \sum_{i=1}^n {\left( \vec X_i - \vec\mu \right) \left( \vec X_i - \vec\mu \right)^T} - \Sigma \right) \Sigma^{-1/2} \right\|_2^2
\leq O(\alpha^2)
\, .
\]
\end{fact}
\begin{proof}

First of all, we restate the left hand side of the inequality in terms of the variance-normalized vectors $\vec V_i$, as follows:
\[
\left\| \Sigma^{-1/2} \left( \frac{1}{n} \sum_{i=1}^n {\left( \vec X_i - \vec\mu \right) \left( \vec X_i - \vec\mu \right)^T} - \Sigma \right) \Sigma^{-1/2} \right\|_2^2
= \left\| \frac{1}{n} \sum_{i=1}^n {\vec V_i \vec V_i^T} - \mathbb{I}_d \right\|_2^2
\, ,
\]
which we can bound with high probability by the classical empirical covariance estimation concentration bounds in \cite{dkk}. More specifically, with probability $1-O(\gamma)$, we have that:
\[
\left\| \frac{1}{n} \sum_{i=1}^n {\vec V_i \vec V_i^T} - \mathbb{I}_d \right\|_2
= \lambda_\text{max}\left( \frac{1}{n} \sum_{i=1}^n {\vec V_i \vec V_i^T} - \mathbb{I}_d \right)
\leq O\left( \sqrt{\frac{d+\log(1/\gamma)}{n}} \right)
\, ,
\]
which directly implies the desired fact.
\end{proof}

\begin{fact}
\label{fact:empirical_mean_feature_vectors}
For every $\alpha > 0$, with probability $1-O(\gamma)$, when
\[
n = \Omega\left( \frac{d+\log(1/\gamma)}{\alpha^2} \right)
\, ,
\]
it holds that
\[
\left\| \frac{1}{n} \sum_{i=1}^n \vec V_i \right\|_2^2
\leq O(\alpha^2)
\, .
\]
\end{fact}
\begin{proof}

The vector sum inside the desired term has a multivariate Gaussian distribution, and we can bound its $\ell_2$-norm with high probability by the classical sub-gaussian concentration bounds in \cite{dkk}. More specifically, with probability $1-O(\gamma)$, we have that:
\[
\left\| \frac{1}{n} \sum_{i=1}^n \vec V_i \right\|_2
\leq O\left( \sqrt{\frac{d+\log(1/\gamma)}{n}} \right)
\, ,
\]
which directly implies the desired fact.
\end{proof}

Directly combining \Cref{lemma:gautam_variant_covariance}, \Cref{lemma:gautam_variant_mean} with \Cref{fact:empirical_covariance} and \Cref{fact:empirical_mean_feature_vectors}, one obtains the stated \Cref{claim:Q1_norm}.

\subsubsection{Proof of \Cref{claim:Q2_norm}}
\label{sec:proofofc3}

Initially, breaking $\vec Q_2$ as
\[
\vec Q_2
= \wh{\vec\mu}' - \frac{1}{n} X^T \vec y
= \left( \wh{\vec\mu}' - \vec\mu' \right) - \left( \frac{1}{n} X^T \vec y - \vec\mu' \right)
\, ,
\]
it follows that
\[
\left\| \Sigma^{-1/2} \vec Q_2 \right\|_2^2
\leq \left\| \Sigma^{-1/2} A'^{-1} \right\|_2^2
\cdot \left\| A' \left( \wh{\vec\mu}' - \vec\mu' \right) \right\|_2^2
+ \left\| \frac{1}{n} \sum_{i=1}^n y_i \vec V_i - \Sigma^{-1/2} \vec\mu' \right\|_2^2
\, ,
\]
where we have again used the notation of the variance-normalized vectors $\vec V_i = \Sigma^{-1/2} \vec X_i$, and the ``private preconditioner'' matrix $A'$ (in this case, obtained for the mean estimation of the random vectors $y_i \vec X_i$), due to \Cref{lemma:preconditioner}, as detailed below (see the first lines of the proof of \Cref{fact:inverse_Sigma'_A'_norm}).

Before stating and proving the facts which lead to the desired result, we need to prove the sub-gaussianity of the vectors $y_i \vec X_i$ that are crucial for the conditions of \Cref{lemma:preconditioner} and the rest of our proof.

\begin{proposition}
[Sub-gaussianity of $y_i \vec X_i$]
\label{prop:subgaussian_yi_Xi}
Let $\vec X_i$ be a random vector sampled according to a multivariate Gaussian distribution with mean value $\vec\mu$ and covariance matrix $\Sigma$ such that $\mathbb{I}_d \preceq \Sigma$, and $y_i$ be a random variable such that $\frac{1}{\rho} \leq |y_i| \leq c$. Then, $y_i \vec X_i$ are sub-gaussian random vectors with covariance matrix $\Sigma'$ such that
\[
\frac{1}{\rho^2} \mathbb{I}_d \preceq \Sigma' \preceq c^2 \Sigma
\, ,
\]
and sub-gaussian norm
\[
\| y_i \vec X_i \|_{\psi_2} \leq c \| \vec X_i \|_{\psi_2}
\, .
\]
\end{proposition}
\begin{proof}
First of all, we have that $\Sigma' = \E\left[ y_i^2 \vec X_i \vec X_i^T \right] \preceq c^2 \Sigma$, since the eigenvalues of $\E\left[ (c^2 - y_i^2) \vec X_i \vec X_i^T \right]$ are non-negative, because the quantity inside the expectation is always a positive semi-definite matrix, and expectation is a linear operator, thus the eigenvalues of the matrix in expectation are also non-negative, by the Courant-Fischer min-max theorem.

Similarly, it can be seen that $\frac{1}{\rho^2} \mathbb{I}_d \preceq \frac{1}{\rho^2} \Sigma \preceq \Sigma'$. We proceed to prove the second part of the proposition, which is the sub-gaussian norm inequality.

We prove the sub-gaussianity of the desired vectors, by \Cref{def:subg}: consider a unit vector $\vec u\in S^{d-1}$, then it holds that
\[
\E\left[ \exp(\lambda^2 y_i^2 \langle \vec X_i, \vec u \rangle^2) \right]
\leq \E\left[ \exp((\lambda c)^2 \langle \vec X_i, \vec u \rangle^2) \right]
\leq \exp(\lambda^2 (cK)^2)
\, ,
\]
for all $\lambda: |\lambda| \leq 1/(Kc)$, where $K$ is the sub-gaussian norm of $\langle \vec X_i, \vec u \rangle$. The sub-gaussian norm follows:
\[
\| y_i \vec X_i \|_{\psi_2} \leq c \| \vec X_i \|_{\psi_2}
\, .
\]
\end{proof}

We are now ready to present the facts that lead to the claim.

\begin{fact}
\label{fact:inverse_Sigma'_A'_norm}
With probability $1-O(\gamma)$, when
\[
n = \Omega\left(
\frac{d^{3/2} \sqrt{\log(\kappa\rho c)} \polylog\left( \frac{d\log(\kappa\rho c)} {\gamma\eps} \right) }{\eps}
\right)
\, ,
\]
it holds that
\[
\left\| \Sigma^{-1/2} A'^{-1} \right\|_2^2
\leq O(c^2)
\, .
\]
\end{fact}
\begin{proof}
According to \Cref{prop:subgaussian_yi_Xi}, the conditions of \Cref{lemma:preconditioner} apply to the variables $\rho y_i \vec X_i$ by a change of variables in the sample complexity of $\kappa' = \rho^2 c^2 \kappa$, where we remind to the reader that $\kappa$ is the largest eigenvalue of the covariance matrix of the feature vectors: $\Sigma \preceq \kappa \mathbb{I}_d$.

Therefore, with
\[
n = \Omega\left(
\frac{d^{3/2} \sqrt{\log(\kappa\rho c)} \polylog\left( \frac{d\log(\kappa\rho c)} {\gamma\eps} \right) }{\eps}
\right)
\, ,
\]
we obtain a matrix $A'$ (which is $\rho$ times the $A$ given by the algorithm of \Cref{lemma:preconditioner} as stated in the previous paragraph) such that with probability $1-O(\gamma)$,
\begin{equation}
\label{eq:A'_Sigma'_A'}
\mathbb{I}_d \preceq A' \Sigma' A' \preceq 1000\mathbb{I}_d
\, .
\end{equation}

Note that the knowledge of $\rho$ is \emph{not} required for \Cref{algo:betahat}, since the change of variables that we did only affects the analysis that we performed here (the privacy of the algorithm is solely based on the upper bound $c$ of the labels, and \emph{not} on $\rho$).

Finally, the desired result holds with probability $1-O(\gamma)$:
\[
\left\| \Sigma^{-1/2} A'^{-1} \right\|_2^2
\leq \left\| \Sigma^{-1/2} \Sigma'^{1/2} \right\|_2^2
\cdot \left\| \Sigma'^{-1/2} A'^{-1} \right\|_2^2
\leq c^2 \cdot 1 = c^2
\, ,
\]
since for the two quantities of interest we have separately the following:

By \Cref{prop:subgaussian_yi_Xi} and properties of the positive semi-definite order, we have that
\begin{align*}
&\Sigma' \preceq c^2 \Sigma \\
\Rightarrow\ &c^2 \Sigma'^{-1} \succeq \Sigma^{-1} \\
\Rightarrow\ &\Sigma^{-1} - c^2 \Sigma'^{-1} \preceq O \\
\Rightarrow\ &\Sigma'^{1/2} \Sigma^{-1} \Sigma'^{1/2} \preceq c^2 \mathbb{I}_d \\
\Rightarrow\ & \left\| \Sigma^{-1/2} \Sigma'^{1/2} \right\|_2^2 \leq c^2
\, .
\end{align*}

At the same time, by \Cref{eq:A'_Sigma'_A'}, we obtain the final term:
\[
\left\| \Sigma'^{-1/2} A'^{-1} \right\|_2^2
= \frac{1}{\sigma_\text{min}^2\left( \Sigma'^{1/2} A' \right)}
= \frac{1}{\lambda_\text{min}\left( \left(\Sigma'^{1/2} A'\right)^T \Sigma'^{1/2} A' \right)}
= \frac{1}{\lambda_\text{min}\left( A' \Sigma' A' \right)}
\leq 1
\, .
\]
\end{proof}

\begin{fact}
\label{fact:subgaussian_yi_Xi_mean}
For every $\eta > 0$, with probability $1-O(\gamma)$, when
\[
n = \Omega\left(
\frac{d \log( \frac{d}{\gamma} )}{\eta^2}
+ \frac{d\polylog( \frac{d\log(1/\delta)} {\eta\gamma\eps} )}{\eta\eps}
+ \frac{\sqrt{d}\log( \frac{d}{\gamma\delta} )}{\eps}
+ \frac{d^{3/2} \sqrt{\log(\kappa\rho c)} \polylog\left( \frac{d\log(\kappa\rho c)} {\gamma\eps} \right)}{\eps}
\right)
\, ,
\]
it holds that
\[
\left\| A' \left( \wh{\vec\mu}' - \vec\mu' \right) \right\|_2^2
\leq O(\eta^2)
\, .
\]
\end{fact}
\begin{proof}
This fact is a direct implication of \Cref{prop:subgaussian_yi_Xi} which guarantees the conditions for \Cref{lemma:gautam_variant_mean} to hold.
\end{proof}

\begin{fact}
\label{fact:empirical_mean_y_X}
For every $\eta > 0$, with probability $1-O(\gamma)$, when
\[
n = \Omega\left(
\frac{d + \log(1/\gamma)}{\eta^2}
\right)
\, ,
\]
it holds that
\[
\left\| \frac{1}{n} \sum_{i=1}^n y_i \vec V_i - \Sigma^{-1/2} \vec\mu' \right\|_2^2
\leq O(\eta^2) \cdot c^2
\, .
\]
\end{fact}
\begin{proof}
We will prove that, under the stated conditions,
\[
\left\| \frac{1}{n} \sum_{i=1}^n \frac{y_i \vec V_i}{c} - \Sigma^{-1/2} \frac{\vec\mu'}{c} \right\|_2
\leq O(\eta)
\, ,
\]
and the result will follow.

First of all, we prove that $\frac{y_i \vec V_i}{c}$ is sub-gaussian and calculate a bound on its sub-gaussian norm. Similarly to the sub-gaussianity of \Cref{prop:subgaussian_yi_Xi}, and by \Cref{lemma:subg_properties}, we have that
\[
\left\| \frac{y_i \vec V_i}{c} \right\|_{\psi_2}
=\frac{ \left\| y_i \vec V_i \right\|_{\psi_2} }{c}
\leq C_1\left\| \vec V_i \right\|_{\psi_2}
\leq C_2
\, ,
\]
for some universal constants $C_1, C_2>0$, since $\vec V_i = \Sigma^{-1/2} \vec X_i$ are variance-normalized random vectors.

Then, noting that
\[
\E\left[ \frac{y_i \vec V_i}{c} \right]
=
\Sigma^{-1/2} \frac{\vec\mu'}{c}
\, ,
\]
and by \Cref{lemma:subg_properties}, it holds that the (centered) quantity $\frac{y_i \vec V_i}{c} - \Sigma^{-1/2} \frac{\vec\mu'}{c}$ is also sub-gaussian with sub-gaussian norm at most a constant times the sub-gaussian norm of the non-centered random vector $\frac{y_i \vec V_i}{c}$. Notice that the covariance matrix of the centered quantity above is $\preceq \mathbb{I}_d$. We will leverage this relationship, alongside the sub-gaussianity of the quantity, to prove the final concentration inequality, from which the fact follows:

\begin{lemma}
\label{lemma:subg-concentr}
There exist universal constants $A, B > 0$ such that, for all $t>0$,
\[
\Pr\left[ \left\| \frac{1}{n} \sum_{i=1}^n \frac{y_i \vec V_i}{c} - \E\left[\frac{y_i \vec V_i}{c} \right] \right\|_2 > t \right]
\leq
4\exp\left( Ad - Bnt^2 \right)
\]
\end{lemma}
\begin{proof}
We denote the covariance matrix of $\frac{y_i \vec V_i}{c}$ as $\Sigma''$, for which it holds that $\Sigma'' \preceq \mathbb{I}_d$, and we also name the variance-normalized random vectors $\left(\Sigma''\right)^{-1/2} \frac{y_i \vec V_i}{c}$ as $\vec W_i$ (therefore, $\E[\vec W_i \vec W_i^T] = \mathbb{I}_d$).

By a classical result of sub-gaussian concentration inequalities (see, e.g., Lemma 2.21 of \cite{dkk}), we have that there exist universal constants $A, B > 0$ such that, for all $t>0$,
\[
\Pr\left[ \left\| \frac{1}{n} \sum_{i=1}^n \vec W_i - \E\left[ \vec W_i \right] \right\|_2 > t \right]
\leq
4\exp\left( Ad - Bnt^2 \right)
\, .
\]

Additionally, by definition of the spectral norm, and since $\Sigma'' \preceq \mathbb{I}_d$, we have that:
\begin{align*}
\left\| \frac{1}{n} \sum_{i=1}^n \vec W_i - \E\left[ \vec W_i \right] \right\|_2
&\geq
\left\| \left( \Sigma'' \right)^{1/2} \right\|_2
\left\| \frac{1}{n} \sum_{i=1}^n \vec W_i - \E\left[ \vec W_i \right] \right\|_2
\\ &\geq
\left\| \left( \Sigma'' \right)^{1/2} \left( \frac{1}{n} \sum_{i=1}^n \vec W_i - \E\left[ \vec W_i \right] \right) \right\|_2
\, ,
\end{align*}
namely, that:
\begin{align*}
\Pr\left[ \left\| \frac{1}{n} \sum_{i=1}^n \frac{y_i \vec V_i}{c} - \E\left[\frac{y_i \vec V_i}{c} \right] \right\|_2 > t \right]
&=
\Pr\left[ \left\| \left( \Sigma'' \right)^{1/2} \left( \frac{1}{n} \sum_{i=1}^n \vec W_i - \E\left[ \vec W_i \right] \right) \right\|_2 > t \right]
\\ &\leq
\Pr\left[ \left\| \frac{1}{n} \sum_{i=1}^n \vec W_i - \E\left[ \vec W_i \right] \right\|_2 > t \right]
\\ &\leq
4\exp\left( Ad - Bnt^2 \right)
\, ,
\end{align*}
as desired.
\end{proof}

Utilizing \Cref{lemma:subg-concentr}, \Cref{fact:empirical_mean_y_X} follows.
\end{proof}

Directly combining \Cref{fact:inverse_Sigma'_A'_norm}, \Cref{fact:subgaussian_yi_Xi_mean}, and \Cref{fact:empirical_mean_y_X}, we obtain the stated \Cref{claim:Q2_norm}.

\section{Proof of \Cref{theorem:main-result2}}
\label{appendix:proof-main-result2}

Again, for convenience, we restate here the (stonger) version of \Cref{theorem:main-result2} that we will prove here:

\begin{theorem}
[Privacy and Accuracy of $\wh{\vec\beta}$ in Private Binary Regression]
Under \Cref{asm:gauss1} with covariance parameter $\kappa$ and \Cref{asm:glm} with true parameter $\vec \beta \in \reals^d$, for every privacy parameters $\eps, \delta > 0$, accuracy parameters $\alpha, \eta > 0$ and confidence $\gamma \in (0,1)$,  \textsc{PrivLearnLSE} (defined in \Cref{algo:betahat}) with $\wh{\vec \mu}_{\vec X} = \vec 0$  is  $(\frac{\eps^2}{2} + \eps\sqrt{2\log(1/\delta)}, \delta)$-differentially private. Moreover, if the number of labeled examples is at least:
\begin{align*}
n &= O\left(
\frac{d \log( \frac{d}{\gamma} )}{\eta^2}
+ \frac{d\polylog( \frac{d\log(1/\delta)} {\eta\gamma\eps} )}{\eta\eps}
+ \frac{d^{3/2} \sqrt{\log\kappa} \polylog\left( \frac{d\log\kappa} {\gamma\eps\delta} \right)}{\eps}
\right)
\\
& + O\left(
\frac{d \log( \frac{d}{\gamma} )}{\alpha^2}
+ \frac{d^{3/2} \polylog\left( \frac{d\log(1/\delta)}{\alpha\gamma\eps} \right)} {\alpha\eps}
\right)
\, ,
\end{align*}
then with probability at least $1-O(\gamma)$ an estimate $\wh{\vec\beta} \in \reals^d$ is successfully output and satisfies
\begin{align*}
\|\wh{\vec\beta}-k\vec\beta\|_2^2
\leq \left\| \wh{\vec w} - k \vec w \right\|_2^2
\leq O\left( \alpha^2 \right)
\cdot \left( 1 + \left\|k \vec w\right\|_2^2 \right)
+ O\left( \eta^2 \right)
\, ,
\end{align*}
where $\wh{\vec w} = \Sigma^{1/2} \wh{\vec\beta}$, $\vec w = \Sigma^{1/2} \vec\beta$ and $k =
\frac{2n}{n-d-1}
\E
\left[ f'\left( \vec\beta^T \vec X_i \right) \right]
\, .$ Finally, \textsc{PrivLearnLSE}
runs in $\poly(n)$ time. 
\end{theorem}
\begin{proof}
The privacy of the algorithm in this Theorem arises directly from the privacy of \Cref{theorem:main-result1}, since the algorithm is the same.

For the accuracy guarantee, as discussed in the Technical overview (\Cref{section:technical-overview}), we first break the norm of the vector difference $\left\| \wh{\vec\beta} - k\vec\beta \right\|_2^2$ into the distance from the estimate $\vec \beta_s^\star$, for which we remind to the reader that we define as
\begin{equation}
\label{eq:similar-beta}
\vec \beta_s^\star =
\left( \frac{1}{n} \sum_{i=n+1}^{2n} \vec X_i \vec X_i^T \right)^{-1} \left( \frac{1}{n} \sum_{i=1}^n y_i \vec X_i \right)
\, ,
\end{equation}
that resembles the Least Squares Estimate but \emph{crucially introduces independence} between the two terms that constitute the Least Squares Estimate (to which we can apply our result from \Cref{theorem:main-result1}), and the distance of the Least-Squares-resembling estimate from a constant multiplicative factor of the true regression coefficient $\vec\beta$ (respectively, of the estimate $\vec w_s^\star$ from $\vec w = \Sigma^{1/2} \vec\beta$), as follows:
\[
\left\| \wh{\vec w} - k\vec w \right\|_2^2
\leq
\left\| \wh{\vec w} - \vec w_s^\star \right\|_2^2
+ \left\| \vec w_s^\star - k \vec w \right\|_2^2
\, .
\]

The first term gets bounded by \Cref{theorem:main-result1}, since as we also noted in the Technical overview (\Cref{section:technical-overview}), the independence between $Q_1$ and $\vec Q_2$ in our proof of \Cref{theorem:main-result1} allows us to prove the same claim for $\vec\beta_s^\star$ as we did for $\vec\beta^\star$ above (see \Cref{appendix:proof-main-result1}). In the remainder of the proof, we focus on bounding the second term.

We first supply the following central Lemma, which uncovers the (unbiased up to a multiplicative factor) relation between the Least-Squares-resembling estimate $\vec\beta_s^\star$ and the true regression coefficient $\vec\beta$, following from Stein's Lemma:

\begin{lemma}
There exists a multiplicative factor $k\in\reals_+$ that depends on the model function $f$, where $f$ as defined in \Cref{asm:glm}, such that the estimate $\vec\beta_s^\star$, as in \Cref{eq:similar-beta}, is an unbiased up to a multiplicative factor estimate of the true parameter $\vec\beta$ of \Cref{asm:glm}, i.e.,
\[
\E\left[ \vec\beta_s^\star \right] = k\vec\beta
\, .
\]
\end{lemma}
\begin{proof}
First, we note the following equality following from the definitions of \Cref{asm:glm}:
\begin{equation}
\label{eq:expe_yi_Xi}
\E\left[ y_i | \vec X_i \right]
= 2f\left( \vec\beta^T \vec X_i \right) - 1
\, .
\end{equation}

Also, by a classical result on Wishart matrices (for instance, see \cite{anderson_multivariate_statistics}), it is true that the inverse sample covariance matrix is proportional to the true covariance matrix for multivariate Gaussian random vectors:
\begin{align}
\label{eq:expe_inv_sample_cov_matrix}
\E\left[ \left( \frac{1}{n} \sum_{i=n+1}^{2n} \vec X_i \vec X_i^T \right)^{-1} \right]
= \frac{n}{n-d-1} \Sigma^{-1}
\, .
\end{align}

By the independence of the first $n$ samples ($1 \dots n$) from the next ($n+1 \dots 2n$), the law of iterated expectations, using \Cref{eq:expe_yi_Xi}, \Cref{eq:expe_inv_sample_cov_matrix} and the zero-mean property of the feature vectors $\vec X_i$, we have that:
\begin{align*}
\E\left[ \vec\beta_s^\star \right]
&= \E\left[ \left( \frac{1}{n} \sum_{i=n+1}^{2n} \vec X_i \vec X_i^T \right)^{-1} \left( \frac{1}{n} \sum_{i=1}^n y_i \vec X_i \right) \right]
\\
&= \E\left[ \left( \frac{1}{n} \sum_{i=n+1}^{2n} \vec X_i \vec X_i^T \right)^{-1} \right]
\E\left[ \frac{1}{n}\sum_{i=1}^n \vec X_i \E[y_i | \vec X_i] \right]
\\
&= \frac{n}{n-d-1} \Sigma^{-1}
\E\left[ \vec X_i \left( 2f\left(\vec\beta^T \vec X_i\right) - 1 \right) \right]
\\
&= \frac{n}{n-d-1} \Sigma^{-1}
\Cov\left[ \vec X_i, ~ 2f\left(\vec\beta^T \vec X_i\right) - 1 \right]
\, ,
\end{align*}

Now, an application of Stein's Lemma (see \Cref{lemma:stein-lemma}), since $\vec X_i$ and $\vec\beta^T \vec X_i$ are jointly Gaussian, suggests that
\[
\Cov\left[ \vec X_i, ~ 2f\left(\vec\beta^T \vec X_i\right) - 1 \right]
= 2 \Cov\left[ \vec X_i, \vec\beta^T \vec X_i \right] \E\left[ f'\left( \vec\beta^T \vec X_i \right) \right]
= 2 \E\left[ f'\left( \vec\beta^T \vec X_i \right) \right] ~ \Sigma ~ \vec\beta
\, ,
\]
and combining with the above equality yields
\[
\E\left[ \vec\beta_s^\star \right]
= k \vec\beta
\, ,
\]
where
\[
k =
\frac{2n}{n-d-1}
\E
\left[ f'\left( \vec\beta^T \vec X_i \right) \right]
\, .
\]
\end{proof}

Continuing to the proof of the result, we use the form as written with the expectation, breaking the term into three sub-terms by adding and subtracting the same quantities (defining the variance-normalized vectors $\vec V_i = \Sigma^{-1/2} \vec X_i$), to deduce that
\begin{align*}
\left\| \vec w_s^\star - k \vec w \right\|_2^2
&= \left\| \vec w_s^\star - \E\left[ \vec w_s^\star \right] \right\|_2^2
\\
&= \left\| \Sigma^{1/2} \left( \frac{1}{n} \sum_{i=n+1}^{2n} \vec X_i \vec X_i^T \right)^{-1} \left( \frac{1}{n} \sum_{i=1}^n y_i \vec X_i \right)
- \frac{n}{n-d-1} \Sigma^{-1/2} \E\left[ y_j \vec X_j \right] \right\|_2^2
\\
&\leq 2\left\| \left( \frac{1}{n} \sum_{i=n+1}^{2n} \vec V_i \vec V_i^T \right)^{-1} \left( \frac{1}{n} \sum_{i=1}^n y_i \vec V_i \right) - \left( \frac{1}{n} \sum_{i=1}^n y_i \vec V_i \right) \right\|_2^2
\\
& + 2\left\| \left( \frac{1}{n} \sum_{i=1}^n y_i \vec V_i \right) - \E\left[ y_j \vec V_j \right] \right\|_2^2 + 2\left\| \frac{d+1}{n-d-1} \E\left[ y_j \vec V_j \right] \right\|_2^2
\\
&= 2\left\| \left( \frac{1}{n} \sum_{i=n+1}^{2n} \vec V_i \vec V_i^T \right)^{-1} \left( \frac{1}{n} \sum_{i=n+1}^{2n} \vec V_i \vec V_i^T - \mathbb{I}_d \right) \left( \frac{1}{n} \sum_{i=1}^n y_i \vec V_i \right) \right\|_2^2
\\
& + 2\left\| \frac{1}{n} \sum_{i=1}^n \left( y_i \vec V_i - \E\left[ y_j \vec V_j \right] \right) \right\|_2^2 + 2\left\| \frac{d+1}{n-d-1} \E\left[ y_j \vec V_j \right] \right\|_2^2
\, .
\end{align*}

When we have at least as many samples as required for \Cref{theorem:main-result1}, it follows from the sub-proofs presented in the proof of this Theorem above (specifically, \Cref{fact:sigma_tilde_inverse_norm}, \Cref{fact:empirical_covariance}, \Cref{fact:empirical_mean_feature_vectors}, and \Cref{fact:empirical_mean_y_X} of \Cref{appendix:proof-main-result1}) that the whole quantity is bounded as follows:
\[
\left\| \vec w_s^\star - \E\left[ \vec w_s^\star \right] \right\|_2^2
\leq O\left( \alpha^2 \right) + O\left( \eta^2 \right)
\, .
\qedhere
\]
\end{proof}

\section{Algorithm and Guarantees on Standard Linear Regression}
\label{app:linear}

We begin this section with a description of \Cref{algo:linear-regression}. First, we can deduce from \Cref{eq:gaussian_error_linear_model} that the (marginal) distribution of labels $y_i$ is a Gaussian distribution $\calN(\vec\beta^T \vec\mu, \vec\beta^T \Sigma \vec\beta + \sigma_\eps^2)$, therefore, defining the vectors $\vec Z_i \in \mathbb{R}^{d+1}$ as
\begin{align}
\label{eq:jointly_Gaussian_rvs}
\vec Z_i = \begin{bmatrix} \vec X_i \\ y_i \end{bmatrix}
\, ,
\end{align}
we can see that they similarly follow a Gaussian distribution, with a covariance matrix $\Sigma'$ that may be written in a block matrix form:
\begin{align}
\label{eq:block_matrix}
\Sigma' = \E\left[ \vec Z_i \vec Z_i^T \right] =
\begin{bmatrix}
\Sigma & \Sigma\vec\beta \\
\vec\beta^T \Sigma & \sigma_\eps^2 + \vec\beta^T \Sigma \vec\beta
\end{bmatrix}
\, .
\end{align}

\Cref{eq:block_matrix} indicates that a natural way to estimate $\vec\beta$ would be to estimate the covariance matrices $\Sigma, \Sigma'$ and then, extracting the last column of $\Sigma'$ (\textsc{ExtractLastColumn} in \Cref{algo:linear-regression}), which (without the last element) is $\Sigma\vec\beta$, left multiply by the inverse of the estimate of $\Sigma$ that we have. Indeed, we show that this approach works and, when the estimation of the above covariance matrices is made according to the differentially private algorithms of Gaussian covariance estimation, the end result is a differentially private estimator of $\vec\beta$ for the ``standard linear regression'' model of \Cref{asm:standard_linear_regression}.

\noindent\textbf{Description of \Cref{algo:linear-regression}.} Having access to the $n$ i.i.d. samples $(\vec X_i, y_i) \in \reals^d \times \reals$, where $\vec X_i \sim \calN(\vec \mu, \Sigma), i \in [n]$, the algorithm initially
computes a differentially private estimate $\wh{\Sigma}$
of the covariance matrix of the $d$-dimensional Gaussian distribution $\calN(\vec \mu, \Sigma)$, using the algorithm $\textsc{LearnGaussian-hd}$, discussed in \Cref{sec:paremestimation}. Then, the algorithm forms the random vectors $\vec Z_i \in\reals^{d+1}$ as in \Cref{eq:jointly_Gaussian_rvs} and computes a differentially private estimate $\wh{\Sigma'}$
of the covariance matrix of the $(d+1)$-dimensional Gaussian distribution with covariance matrix of the block form of \Cref{eq:block_matrix}, again using the algorithm $\textsc{LearnGaussian-hd}$. From the matrix $\wh{\Sigma'}$, the algorithm obtains only the first $d$ elements of the last column of that matrix, naming them as $\wh{\Sigma\vec\beta} \in\reals^d$ (hinting at the form of \Cref{eq:block_matrix}).
Armed with these estimates, the differentially private estimate of $\vec\beta$ is finally given by:
\begin{align}
\label{eq:standard_lin_regression_dpestimate}
\wh{\vec\beta} = \wh{\Sigma}^{-1} \wh{\Sigma\vec\beta}
\, ,
\end{align}
whose privacy follows from appropriate composition rules. We now proceed to prove the accuracy guarantee of this estimate.

\begin{algorithm}[!t] 
\caption{Private Estimation of Linear Regression Coefficient.}
\label{algo:linear-regression}
\begin{algorithmic}[1]
\State \textbf{Input:} $(X, \vec y) = (\vec X_i, y_i)_{i \in [n]}$ with $\vec X_i \sim \calN(\vec \mu, \Sigma)$, where $\vec \mu, \Sigma$ are unknown and $n$ satisfies \Cref{theorem:main-result3}.
\State \textbf{Parameters:} Privacy $\eps, \delta > 0$, accuracy $\alpha, \eta > 0$, confidence $\gamma \in (0,1)$,
covariance spectral norm bound $\kappa$. 
\State \textbf{Output:} Estimate $\wh{\vec\beta}$ that approaches the true vector $\vec\beta$ in $L_2$ norm with high probability.
\vspace{1mm}
\Procedure{PrivLearnLinear}{$(X, \vec y), \eps, \delta, \alpha, \eta, \gamma, \kappa$} \State Set $\vec Z_i \gets [\vec X_i, y_i]^T$ for $i \in [n]$
\State $L \gets \{ \Theta(\eps), \Theta(\delta), \Theta(\alpha), \gamma, \kappa \}$
\State $\wh{\Sigma'} \gets \textsc{LearnGaussian-hd}(\{\vec Z_i\}_{i
\in [n]}, L)$

\State $[\wh{\Sigma \vec \beta},~\wh{\sigma}] \gets \textsc{ExtractLastColumn}(\wh{\Sigma'})$ \Comment{\emph{See \Cref{eq:block_matrix-main},} $\wh{\Sigma \vec \beta} \in \reals^d, \wh{\sigma} = \wh{\sigma_\eps^2 + \vec\beta^T \Sigma \vec\beta} \in \reals$.}
\State Draw $\vec X_i, i \in \{n+1,\ldots, 2n\}$ from $\calN(\vec \mu, \Sigma)$
\State $\wh{\Sigma} \gets \textsc{LearnGaussian-hd}(\{\vec X_i\}_{i \in [n+1..2n]}, L)$
\State \textbf{if} $\wh{\Sigma}$
is not invertible\footnotemark\ 
\textbf{then} Output $\perp$
\State Output the private estimate $\wh{\vec\beta} = 
\wh{\Sigma}^{-1} \wh{\Sigma \vec \beta} $
\EndProcedure
\end{algorithmic}
\end{algorithm}
\footnotetext{The invertibility of the matrix in \textbf{Line 10} holds with high probability and the non-invertibility bad event is captured by the $O(\gamma)$ failure probability of \Cref{theorem:accuracy_standard_lin_regression}.}

\begin{theorem}
[Accuracy of $\wh{\vec\beta}$ in Private Standard Linear Regression]
\label{theorem:accuracy_standard_lin_regression}
Under \Cref{asm:standard_linear_regression} with parameters $(\kappa, 
\Sigma')$ where $\Sigma'$ is defined as in \Cref{eq:block_matrix}, for all privacy parameters $\eps, \delta > 0$, accuracy parameters $\alpha, \eta > 0$ and confidence $\gamma \in (0,1)$, there exists an algorithm (\Cref{algo:linear-regression}) that is $(\frac{\eps^2}{2} + \eps\sqrt{2\log(1/\delta)}, \delta)$-differentially private, and if the number of labeled examples is at least: 
\begin{align*}
n &= O\left(
\frac{d + \log(1/\gamma)}{\eta^2}
+ \frac{d^{3/2} \polylog\left( \frac{d}{\eta\gamma\eps} \right)} {\eta\eps}
+ \frac{d^{3/2} \sqrt{\log\left( \kappa(\Sigma') \right)} \polylog\left( \frac{d\log\left( \kappa(\Sigma') \right)} {\gamma\eps} \right) }{\eps}
\right)
\\
& + O\left(
\frac{d + \log(1/\gamma)}{\alpha^2}
+ \frac{d^{3/2} \polylog\left( \frac{d}{\alpha\gamma\eps} \right)} {\alpha\eps}
\right)
\, ,
\end{align*}
then with probability at least $1-O(\gamma)$ an estimate $\wh{\vec\beta} \in \reals^d$ is successfully output and along with the ``true'' regression coefficient $\vec\beta$ satisfies:
\begin{align}
\label{eq:standard_lin_regression_accuracy_guarantee}
\left\| \wh{\vec\beta} - \vec\beta \right\|_2^2
\leq \left\| \wh{\vec w} - \vec w \right\|_2^2
\leq O\left( \alpha^2 \right)
\cdot \left\|\vec w\right\|_2^2
+ O\left( \eta^2 \right) \cdot \lambda_\text{max}^2(\Sigma')
\, ,
\end{align}
where $\kappa(\Sigma') = \frac{\lambda_\text{max}(\Sigma')}{\lambda_\text{min}(\Sigma')}$ is the condition number of the block matrix $\Sigma'$ as in \Cref{eq:block_matrix}, $\wh{\vec w} = \Sigma^{1/2} \wh{\vec\beta} ~\text{and}~ \vec w = \Sigma^{1/2} \vec\beta$.
Finally, the algorithm runs in $\poly(n)$ time. 
\end{theorem}
\begin{proof}
The privacy of the algorithm in this Theorem arises directly from the privacy of the differentially private covariance estimation algorithm and the composition theorems.

For the accuracy guarantee, we begin by adding and subtracting the quantities of each factor of $\wh{\vec\beta} \in \reals^d$, as follows:
\[
\wh{\vec\beta} - \vec\beta = \left( \wh{\Sigma}^{-1} - \Sigma^{-1} \right) \Sigma\vec\beta + \wh{\Sigma}^{-1} \left( \wh{\Sigma\vec\beta} - \Sigma\vec\beta \right)
\, .
\]

Then, substituting the quantities $\vec w = \Sigma^{1/2} \vec\beta$ and left multiplying both sides of the equation by $\Sigma^{1/2}$, we obtain that
\begin{align*}
& \wh{\vec w} - \vec w = \left( \Sigma^{1/2} \wh{\Sigma}^{-1} \Sigma^{1/2} - \mathbb{I}_d \right) \vec w + \Sigma^{1/2} \wh{\Sigma}^{-1} \left( \wh{\Sigma\vec\beta} - \Sigma\vec\beta \right) \\
\Leftrightarrow\ & \wh{\vec w} - \vec w = \left( \Sigma^{1/2} \wh{\Sigma}^{-1} \Sigma^{1/2} \right) \left[ - \left( \Sigma^{-1/2} \wh{\Sigma} \Sigma^{-1/2} - \mathbb{I}_d \right) \vec w + \Sigma^{-1/2} \left( \wh{\Sigma\vec\beta} - \Sigma\vec\beta \right) \right]
\, .
\end{align*}

Using Cauchy-Schwartz and the sub-multiplicative property of the spectral norm, we establish the following inequality:
\[
\left\| \wh{\vec w} - \vec w \right\|_2^2 \leq 2 \left\| \Sigma^{1/2} \wh{\Sigma}^{-1} \Sigma^{1/2} \right\|_2^2 \left( \left\| \Sigma^{-1/2} \left( \wh{\Sigma} - \Sigma \right) \Sigma^{-1/2} \right\|_2^2 \cdot \left\| \vec w \right\|_2^2 + \left\| \Sigma^{-1/2} \left( \wh{\Sigma\vec\beta} - \Sigma\vec\beta \right) \right\|_2^2 \right)
\, .
\]

In order to bound the constituent terms of the right-hand-side of this inequality, we need a very similar claim to \Cref{claim:inverse_sigma_norm} which we state below without proof (since it is almost the same as \Cref{sec:proofofc1}), \Cref{lemma:gautam_variant_covariance} (which we remind to the reader that it is applied to the $2n$ sample differences $\frac{1}{\sqrt{2}}\left( \vec X_{2i} - \vec X_{2i-1} \right)$, so that they have zero mean), and \Cref{claim:last_column_norm}, which is the main subject that we elaborate on in \Cref{subsec:proof_of_last_column_norm}.

\begin{claim}
[Similar to \Cref{claim:inverse_sigma_norm}]
\label{claim:similar_inverse_sigma_norm}
When
\[
n = \Omega\left(
d
+ \log(1/\gamma)
+ \frac{d^{3/2} \sqrt{\log\kappa} \polylog\left( \frac{d\log\kappa} {\gamma\eps} \right) }{\eps}
\right)
\, ,
\]
the following inequality holds with probability $1-O(\gamma)$: \[
\left\| \Sigma^{1/2} \wh{\Sigma}^{-1} \Sigma^{1/2} \right\|_2^2
\leq
O\left( 1 \right)
\, .
\]
\end{claim}

\begin{claim}
\label{claim:last_column_norm}
When
\[
n = \Omega\left(
\frac{d + \log(1/\gamma)}{\eta^2}
+ \frac{d^{3/2} \polylog\left( \frac{d}{\eta\gamma\eps} \right)} {\eta\eps}
+ \frac{d^{3/2} \sqrt{\log\left( \kappa(\Sigma') \right)} \polylog\left( \frac{d\log\left( \kappa(\Sigma') \right)} {\gamma\eps} \right) }{\eps}
\right)
\, ,
\]
the following inequality holds with probability $1-O(\gamma)$: \[
\left\| \Sigma^{-1/2} \left( \wh{\Sigma\vec\beta} - \Sigma\vec\beta \right) \right\|_2^2
\leq
O\left( \eta^2 \right) \cdot \lambda_\text{max}^2(\Sigma')
\, .
\]
\end{claim}

Combining \Cref{claim:similar_inverse_sigma_norm}, \Cref{lemma:gautam_variant_covariance}, and \Cref{claim:last_column_norm} with a union bound of the respective events, we directly obtain \Cref{theorem:accuracy_standard_lin_regression}, since
\[
\left\| \wh{\vec\beta} - \vec\beta \right\|_2^2
= \left\| \Sigma^{-1/2} \Sigma^{1/2} \left( \wh{\vec\beta} - \vec\beta \right) \right\|_2^2
\leq \left\| \Sigma^{-1/2} \right\|_2^2
\cdot \left\| \wh{\vec w} - \vec w \right\|_2^2
\leq \left\| \wh{\vec w} - \vec w \right\|_2^2
\, ,
\]
by the sub-multiplicative property of the norm and because $\mathbb{I}_d \preceq \Sigma$.
\end{proof}

\subsection{Proof of \Cref{claim:last_column_norm}}
\label{subsec:proof_of_last_column_norm}

First of all, we note that, according to the first lines of the proof of \Cref{fact:inverse_Sigma'_A'_norm}, in order to apply the (accuracy) results of \Cref{lemma:preconditioner} and \Cref{lemma:gautam_variant_covariance} to the covariance estimation of the random vectors $\vec Z_i \in\reals^{d+1}$ as in \Cref{eq:jointly_Gaussian_rvs}, a change of variables is needed, that affects solely the analysis of the algorithm (and more specifically, appears in a change of the sample complexity). Therefore, the specific result which applies in our case here is stated in the following fact.

\begin{fact}
[Covariance $\wh{\Sigma}'$ estimation accuracy]
For every $\eta > 0$, the output $\wh{\Sigma}'$ of algorithm $\textsc{LearnGaussian-hd}$ when given at least $n$ samples $\vec Z_i$ with
\[
n = O\left(
\frac{d + \log(1/\gamma)}{\eta^2}
+ \frac{d^{3/2} \polylog\left( \frac{d}{\eta\gamma\eps} \right)} {\eta\eps}
+ \frac{d^{3/2} \sqrt{\log\left( \kappa(\Sigma') \right)} \polylog\left( \frac{d\log\left( \kappa(\Sigma') \right)} {\gamma\eps} \right) }{\eps}
\right)
\, ,
\]
where $\kappa(\Sigma') = \frac{\lambda_\text{max}(\Sigma')}{\lambda_\text{min}(\Sigma')}$ is the condition number of the block matrix $\Sigma'$ as in \Cref{eq:block_matrix}, satisfies the following accuracy guarantee with probability $1-O(\gamma)$:
\begin{align}
\label{eq:Sigma'_accuracy_standard_lin_regression}
\left\| \Sigma'^{-1/2} \left( \wh{\Sigma}' - \Sigma' \right) \Sigma'^{-1/2} \right\|_2 \leq O(\eta)
\, .
\end{align}
\end{fact}

We remind to the reader the form of the block matrix $\Sigma'$ which is as follows:
\begin{align}
\label{eq:block_matrix_again}
\Sigma' =
\begin{bmatrix}
\Sigma & \Sigma\vec\beta \\
\vec\beta^T \Sigma & \sigma_\eps^2 + \vec\beta^T \Sigma \vec\beta
\end{bmatrix}
\, .
\end{align}

By definition of the spectral norm, from \Cref{eq:Sigma'_accuracy_standard_lin_regression} we have that for every vector $\vec u\in\reals^{d+1} : \|\vec u\|_2 \leq 1$, it holds that $ \left\| \Sigma'^{-1/2} \left( \wh{\Sigma}' - \Sigma' \right) \Sigma'^{-1/2} \vec u \right\|_2 \leq O(\eta) $. Taking advantage of the spectral decomposition of the (real symmetric, PSD) matrix $\Sigma' = U \Lambda U^T$ for some unitary orthogonal matrix $U$ and diagonal matrix $\Lambda$, it is well-known that $\Sigma'^{-1/2} = \Lambda^{-1/2} U^T$, and because $\| U^T \vec u \|_2 = \| \vec u \|_2$ for every vector $\vec u\in\reals^{d+1}$ (since $U$ is an orthonormal matrix), and since $\sqrt{\lambda_\text{max}(\Sigma')} \Lambda^{-1/2} \succeq \mathbb{I}_{d+1}$, we conclude that by choosing $\vec u\in\reals^{d+1} : \sqrt{\lambda_\text{max}(\Sigma')} \Lambda^{-1/2} \vec u = \vec e_{d+1}$ (where $\vec e_{d+1}$ is the unit vector that has only the $(d+1)$-th coordinate $1$ and all other coordinates $0$), it is true that $ \left\| \Sigma'^{-1/2} \left( \wh{\vec v}_{d+1} - \vec v_{d+1} \right) \right\|_2^2 \leq O(\eta^2) \cdot \lambda_\text{max}(\Sigma') $, where $\wh{\vec v}_{d+1}, \vec v_{d+1}$ are the last columns of the matrices $\wh{\Sigma}'$ and $\Sigma'$ respectively (see \Cref{eq:block_matrix_again} for what the last column looks like). Therefore, it is immediate that
\begin{align*}
\left\| \wh{\vec v}_{d+1} - \vec v_{d+1} \right\|_2^2
&=
\left\| \Sigma'^{1/2} \Sigma'^{-1/2} \left( \wh{\vec v}_{d+1} - \vec v_{d+1} \right) \right\|_2^2
\\ &\leq
\left\| \Sigma'^{1/2} \right\|_2^2 \cdot \left\| \Sigma'^{-1/2} \left( \wh{\vec v}_{d+1} - \vec v_{d+1} \right) \right\|_2^2
\\ &\leq
O(\eta^2) \cdot \lambda_\text{max}^2(\Sigma')
\, ,
\end{align*}
and since the first $d$ coordinates of $\wh{\vec v}_{d+1} - \vec v_{d+1} \in\reals^{d+1}$ are the vector $\wh{\Sigma\vec\beta} - \Sigma\vec\beta \in\reals^d$ (see the structure of \Cref{eq:block_matrix_again}), we have that $\left\| \wh{\Sigma\vec\beta} - \Sigma\vec\beta \right\|_2^2 \leq \left\| \wh{\vec v}_{d+1} - \vec v_{d+1} \right\|_2^2 \leq O(\eta^2) \cdot \lambda_\text{max}^2(\Sigma')$.

\Cref{claim:last_column_norm} follows, since
\[
\left\| \Sigma^{-1/2} \left( \wh{\Sigma\vec\beta} - \Sigma\vec\beta \right) \right\|_2^2
\leq
\left\| \Sigma^{-1/2} \right\|_2^2
\cdot
\left\| \wh{\Sigma\vec\beta} - \Sigma\vec\beta \right\|_2^2
\leq
\left\| \wh{\Sigma\vec\beta} - \Sigma\vec\beta \right\|_2^2
\leq
O(\eta^2) \cdot \lambda_\text{max}^2(\Sigma')
\, ,
\]
by the sub-multiplicative property of the norm and because $\mathbb{I}_d \preceq \Sigma$.

The proof has now been completed. Of course, the condition number $\kappa(\Sigma')$ is an interesting quantity that merits consideration to examine what it depends upon.
From Theorem 1 of \cite{loose_bounds_block_matrix}, one may deduce the following upper bound on the largest eigenvalue of $\Sigma'$:
\begin{align}
\label{eq:lambda_max_block_matrix}
\lambda_\text{max}(\Sigma')
\leq
2\left( \vec\beta^T \Sigma \vec\beta + \text{max}\left( \kappa, \sigma_\eps^2 \right) \right)
\, ,
\end{align}
where we remind to the reader that $\kappa = \lambda_\text{max}(\Sigma)$.

The lower bound on the smallest eigenvalue provided by the above work \citep{loose_bounds_block_matrix} is non-optimal, since it may be negative at certain cases, while the matrix itself only ever exhibits non-negative eigenvalues (since it is PSD, by definition of being a covariance matrix). A better bound may be deduced by block matrix eigenvalue approaches \citep{lower_upper_bounds_block_matrix}, as follows:
\begin{align}
\lambda_\text{min}(\Sigma')
&\geq
\frac{\sigma_\eps^2 + \vec\beta^T \Sigma \vec\beta + \lambda_\text{min}(\Sigma)}{2}
- \sqrt{
\left( \frac{\sigma_\eps^2 + \vec\beta^T \Sigma \vec\beta + \lambda_\text{min}(\Sigma)}{2} \right)^2
- \sigma_\eps^2 \lambda_\text{min}(\Sigma)
}
\nonumber \\ &=
\frac
{ \sigma_\eps^2 \lambda_\text{min}(\Sigma) }
{ \frac{1}{2} \left(
\sigma_\eps^2 + \vec\beta^T \Sigma \vec\beta + \lambda_\text{min}(\Sigma)
+ \sqrt{
\left( \sigma_\eps^2 + \vec\beta^T \Sigma \vec\beta + \lambda_\text{min}(\Sigma) \right)^2
- 4 \sigma_\eps^2 \lambda_\text{min}(\Sigma)
}
\right) }
\nonumber \\ &\geq
\frac
{ \sigma_\eps^2 \lambda_\text{min}(\Sigma) }
{ \sigma_\eps^2 + \vec\beta^T \Sigma \vec\beta + \lambda_\text{min}(\Sigma) }
\label{eq:lambda_min_block_matrix}
\, .
\end{align}

The first form of the lower bound given above is tight, as we will now prove by examining the specific case of $\Sigma = \kappa \mathbb{I}_d$. In this case, one would have that:
\[
\Sigma' =
\begin{bmatrix}
\kappa \mathbb{I}_d & \kappa\vec\beta \\
\kappa \vec\beta^T & \sigma_\eps^2 + \kappa \|\vec\beta\|_2^2
\end{bmatrix} = \kappa
\begin{bmatrix}
\mathbb{I}_d & \vec\beta \\
\vec\beta^T & \frac{\sigma_\eps^2}{\kappa} + \|\vec\beta\|_2^2
\end{bmatrix}
\, ,
\]
reducing our calculations to the simple case of covariance matrix $\mathbb{I}_d$, for which the second matrix written in the above equation has eigenvalues $t$ according to the roots of the equation
\begin{align*}
&(1-t)^d \left( \frac{\sigma_\eps^2}{\kappa} + \|\vec\beta\|_2^2 - t \right)
- \sum_{i=1}^d \beta_i^2 (1-t)^{d-1} = 0
\\ \Leftrightarrow \ \ 
&(1-t)^{d-1}
\left(
t^2
- \left( 1 + \|\vec\beta\|_2^2 + \frac{\sigma_\eps^2}{\kappa} \right) t
+ \frac{\sigma_\eps^2}{\kappa} \right)
= 0
\, ,
\end{align*}
where $\beta_i$ is the $i$-th coordinate of the vector $\vec\beta$. Therefore $\Sigma'$ has the following smallest eigenvalue (the smallest of the two roots of the quadratic equation, which is guaranteed to be $\leq 1$, i.e., smaller than the other eigenvalues):
\begin{align*}
\lambda_\text{min}(\Sigma')
&= \frac{\kappa}{2}
\left(
1 + \|\vec\beta\|_2^2 + \frac{\sigma_\eps^2}{\kappa}
- \sqrt{
\left( 1 + \|\vec\beta\|_2^2 + \frac{\sigma_\eps^2}{\kappa} \right)^2
- 4 \frac{\sigma_\eps^2}{\kappa}
}
\right)
\\ &=
\frac{\sigma_\eps^2 + \kappa \|\vec\beta\|_2^2 + \kappa}{2}
- \sqrt{
\left( \frac{\sigma_\eps^2 + \kappa \|\vec\beta\|_2^2 + \kappa}{2} \right)^2
- \kappa \sigma_\eps^2
}
\, ,
\end{align*}
which neatly matches the first form of the lower bound given in \Cref{eq:lambda_min_block_matrix}, since $\Sigma = \kappa \mathbb{I}_d$.

Therefore, the condition number $\kappa(\Sigma')$ (which is at most the ratio of the right-hand-sides of \Cref{eq:lambda_max_block_matrix} to \Cref{eq:lambda_min_block_matrix}) and the largest eigenvalue $\lambda_\text{max}(\Sigma')$ depend on both $\vec\beta$ and $\sigma_\eps^2$ besides the usual dependence on the smallest and largest eigenvalues of the feature vector covariance matrix $\Sigma$, i.e., $\lambda_\text{max}(\Sigma) \leq \kappa$ and $\lambda_\text{min}(\Sigma) \geq 1$ respectively. We note, in particular, that this means that, when $\| \vec\beta \|_2$ is large, more samples will be necessary to achieve a fixed additive accuracy, as indicated by \Cref{eq:standard_lin_regression_accuracy_guarantee}.



\end{document}